\documentclass[letterpaper,twocolumn,10pt]{article}
\usepackage{usenix2025_SOUPS}

\usepackage{tikz}
\usepackage{amsmath}
\usepackage{amssymb,amsthm}
\usepackage{graphicx}
\usepackage{caption}
\usepackage{subcaption}
\usepackage{booktabs}
\usepackage{enumerate} 
\usepackage{enumitem} 
\usepackage{multirow} 
\usepackage{comment}
\usepackage{balance}
\usepackage{stfloats}
\usepackage{breakurl}
\usepackage{xurl}
\usepackage{authblk}

\newcommand{\mypara}[1]{\smallskip\noindent\textbf{#1:}}

\newtheorem{prop}{Proposition}

\begin{document}
\pagenumbering{gobble} 
\date{}

\title{\Large \bf Design and Evaluation of Privacy-Preserving Protocols for \\ Agent-Facilitated 
Mobile Money Services in Kenya}

\def\plainauthor{Author name(s) for PDF metadata. Don't forget to anonymize for submission!}


\author[1]{
{\rm Karen Sowon}} 
\author[2]{
{\rm Collins W. Munyendo}}
\author[3]{
{\rm Lily Klucinec}}
\author[4]{
{\rm Eunice Maingi}}
\author[4]{
{\rm Gerald Suleh}}
\author[3]{
{\rm \\ Lorrie Faith Cranor}}
\author[3]{
{\rm Giulia Fanti}}
\author[5]{
{\rm Conrad Tucker}}
\author[5]{
{\rm Assane Gueye}}

\affil[1]{\emph{Indiana University}}
\affil[2]{\emph{The George Washington University}}
\affil[3]{\emph{Carnegie Mellon University}}
\affil[4]{\emph{Strathmore University}}
\affil[5]{\emph{Carnegie Mellon University-Africa}}

\maketitle
\thecopyright

\begin{abstract} 
Mobile Money (MoMo), a technology that allows users to complete financial transactions using a mobile phone without requiring a bank account, is a common method for processing financial transactions in Africa and other developing regions. Users can deposit 
and withdraw money with the help of human agents. 
During deposit and withdraw operations, know-your-customer (KYC) processes require agents to access and verify customer information such as name and ID number, which can introduce privacy and security risks. In this work, we design alternative protocols for MoMo deposits/withdrawals that protect users' privacy while enabling KYC checks by redirecting the flow of sensitive information from the agent to the MoMo provider. 
We evaluate the usability and efficiency of our proposed protocols in a role-play and semi-structured interview study with 32 users and 15 agents in Kenya. We find that users and agents 
prefer the new protocols, due in part to convenient and efficient verification using biometrics as well as better data privacy and access control. However, our study also surfaced challenges that need to be addressed before these protocols can be deployed. 
\end{abstract}

\newcommand{\figurePrivacyDelegated}[0]{
    \begin{figure*}[h!]
    \centering
    \vspace{-2cm}
    \fbox{%
    \begin{minipage}{0.9\textwidth}
    \begin{subfigure}[b]{0.45\linewidth}
        \centering
        \includegraphics[height=9cm, width=8.4cm]{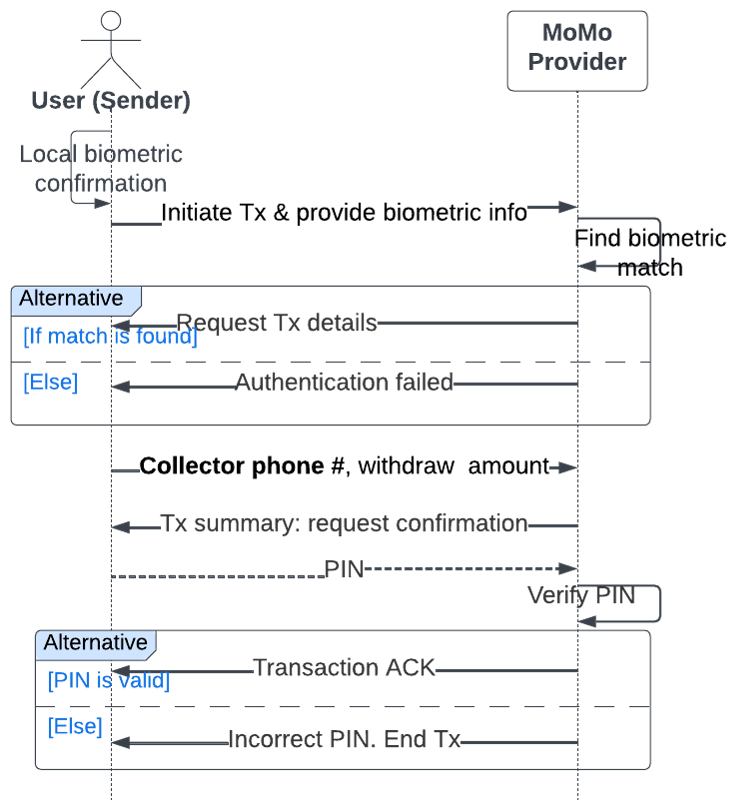}
        \caption{Sender's side}
        \label{fig:sender_side}
    \end{subfigure}
    \hfill
    \begin{subfigure}[b]{0.49\linewidth}
        \centering
        \includegraphics[width=\linewidth]{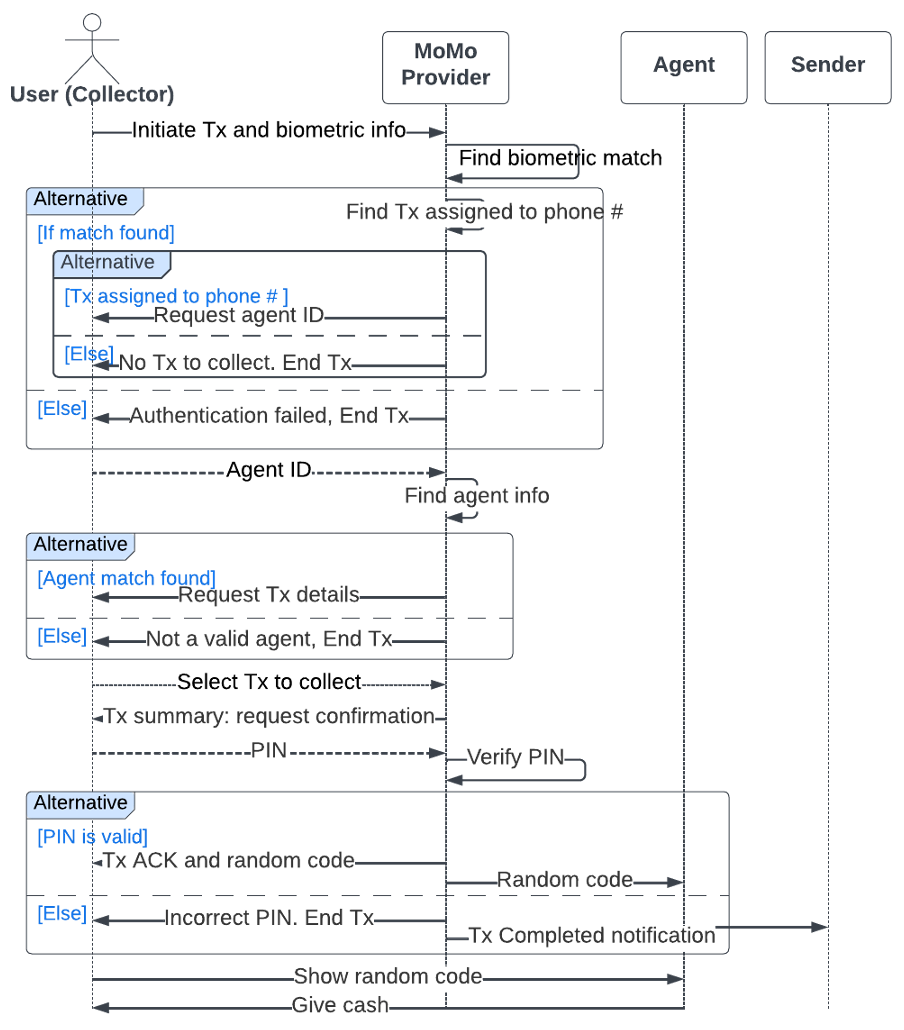}
        \caption{Collector's side}
        \label{fig:collector_side}
    \end{subfigure}
    \end{minipage}
    }
    \caption{Collector's and sender's side of the privacy-preserving delegated withdrawal 
    }
    \label{fig:delegated_withdraw}
\end{figure*}
}

\newcommand{\figureCashOut}[0]{
    \begin{figure*}[t]
    \centering
    \vspace{-1cm}
    \fbox{%
    \begin{minipage}{0.97\textwidth}
    \begin{subfigure}[b]{0.47\textwidth}
        \centering
        \includegraphics[height=8cm, width=8.9cm]{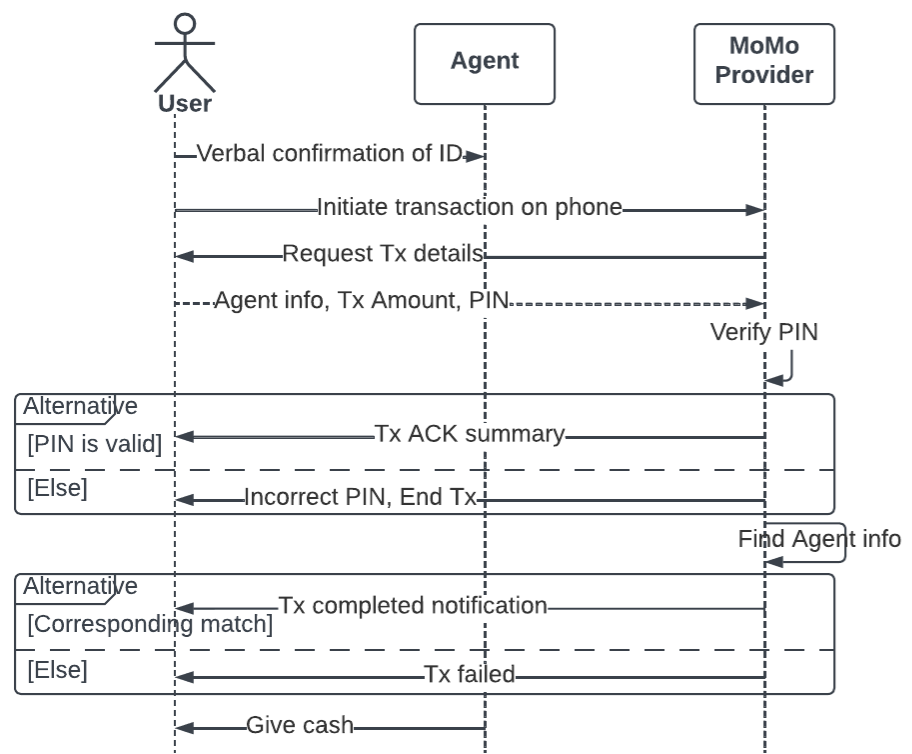}
        \caption{Current withdrawal process}
        \label{fig:curr_cashout}
    \end{subfigure}
    \hfill
    \begin{subfigure}[b]{0.50\textwidth}
        \centering
        \includegraphics[width=\linewidth]{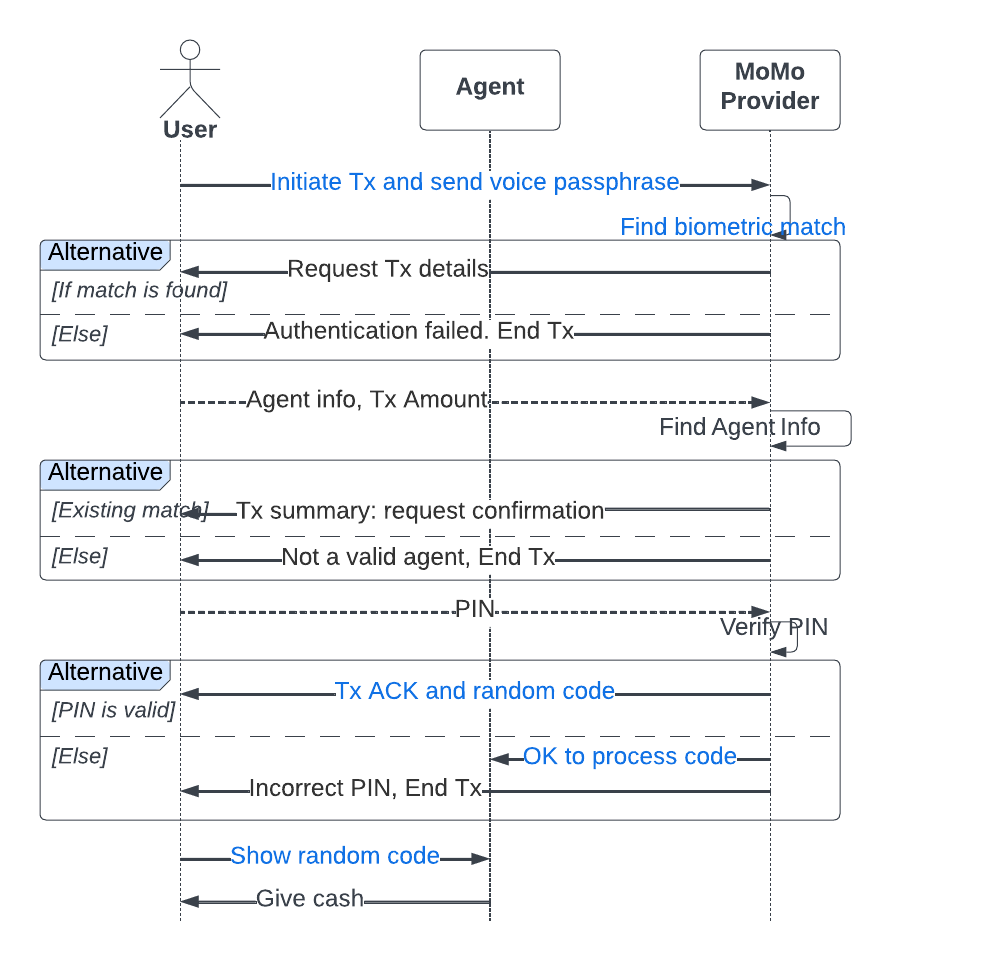}
        \caption{Privacy-preserving withdrawal process for feature phone users}
        \label{fig:priv_cashout}
    \end{subfigure}
    \end{minipage}
    }
    \caption{Withdrawal under the current protocol (left) and the proposed privacy-preserving protocol (right)}
    \label{fig:cashout}
\end{figure*}
}

\newcommand{\figurePrivacyDeposit}[0]{
\begin{figure}[h!]
\includegraphics[width=\linewidth]{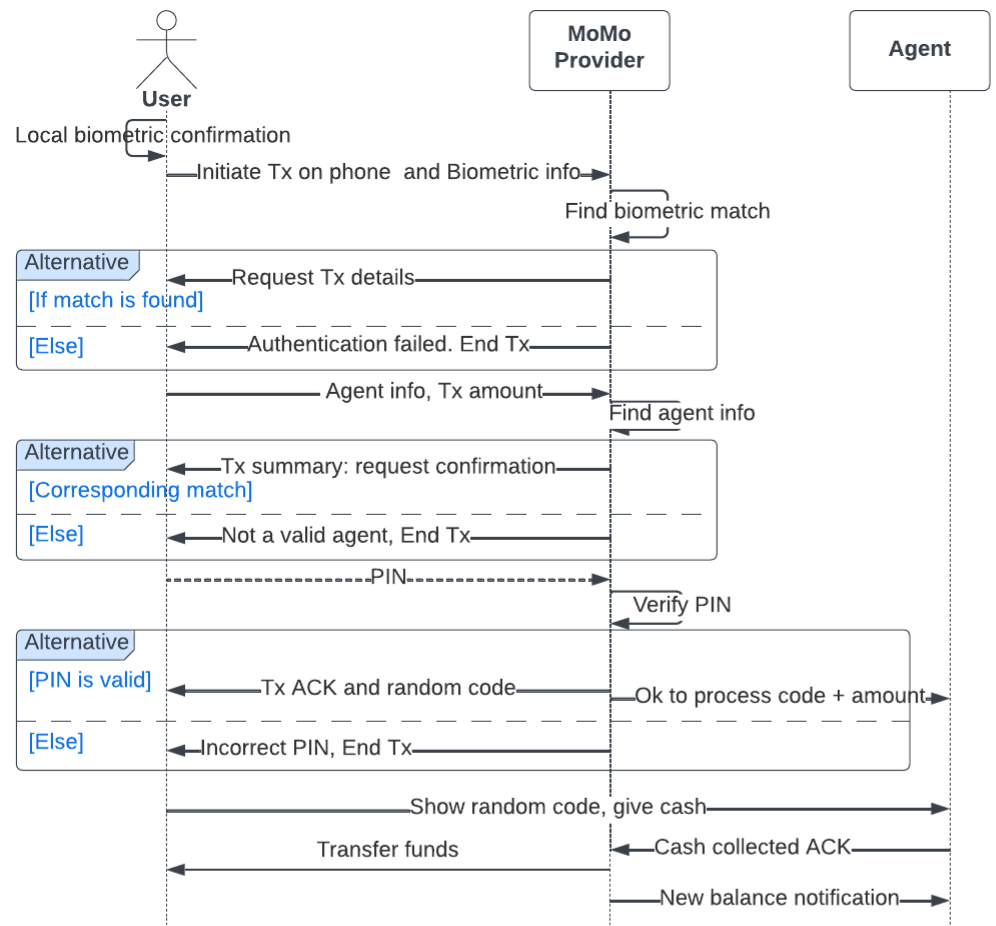}
\caption{Cashing-in with the privacy-preserving process for smart phone users}
\label{fig:priv_cashin}
\end{figure}
}

\newcommand{\figureDeposit}[0]{
    \begin{figure*}[h!]
    \centering
    \vspace{-1cm}
    \fbox{%
    \begin{minipage}{0.9\textwidth}
    \begin{subfigure}[b]{0.47\linewidth}
        \centering
        \includegraphics[width=\linewidth]{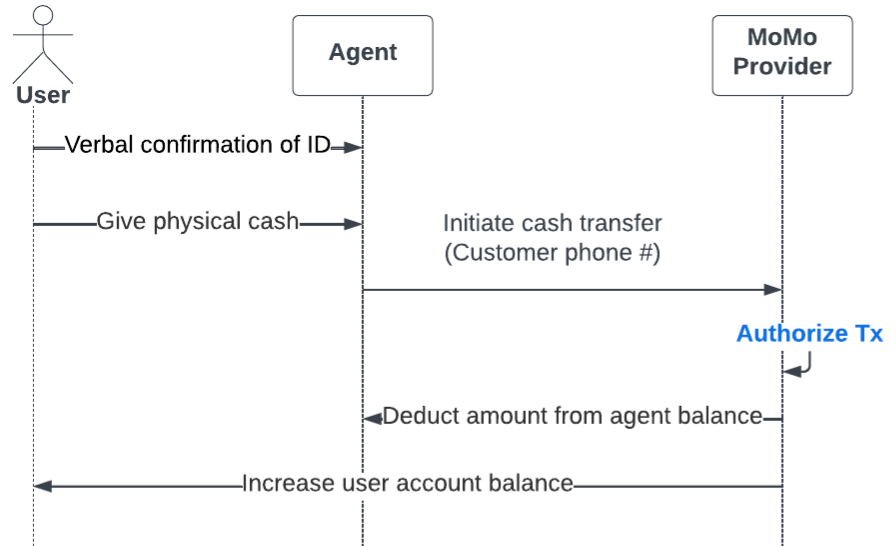}
        \caption{Current deposit}
        \label{fig:curr_dep}
    \end{subfigure}
    \hfill
    \begin{subfigure}[b]{0.47\linewidth}
        \centering
        \includegraphics[width=\linewidth]{Figures/PP-Deposit.png}
        \caption{Privacy-preserving deposit}
        \label{fig:priv_dep}
    \end{subfigure}
    \end{minipage}
    }
    \caption{Current and privacy-preserving techniques for depositing money}
    \label{fig:deposit}
\end{figure*}
}

\newcommand{\figureCurrentProxy}[0]{
\begin{figure}[t]
  \fbox{\includegraphics[width=0.97\linewidth,scale=0.5]{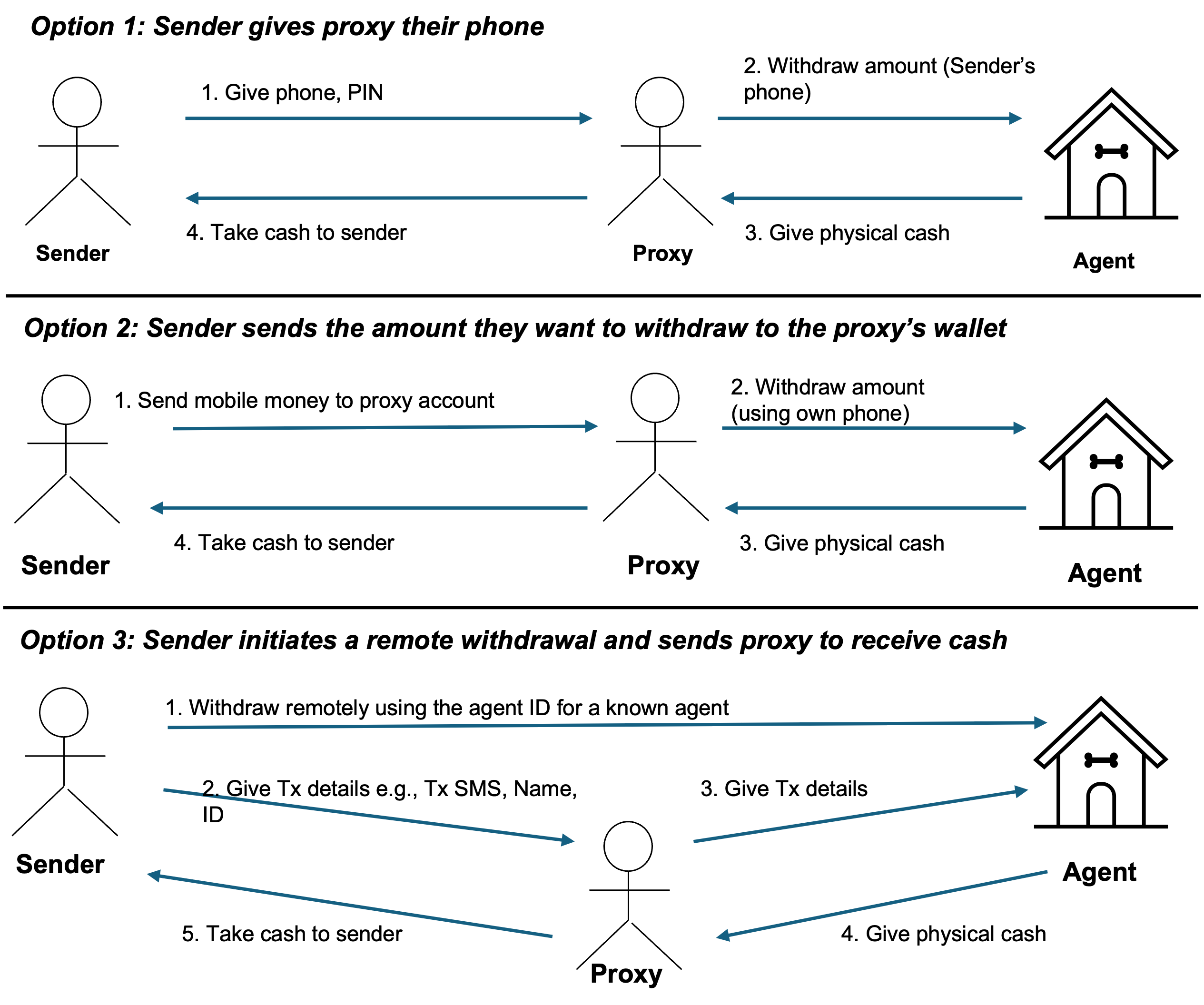}}
  \caption{Current workarounds for proxy withdrawal}
  \label{current_proxy}
\end{figure}
}

\section{Introduction}
Mobile money (MoMo) is a technology that allows mobile phone users to exchange and store money with their phones and associated phone numbers, without the need for a traditional financial account~\cite{donovan2012mobile,jack2011mobile}. MoMo has revolutionized digital transactions in many emerging economies~\cite{ahmad2020mobile}, e.g., sub-Saharan Africa. 
To safeguard MoMo transactions, regulations in most jurisdictions mandate that financial institutions acquire proof of customer identity before transacting, a process known as ``know your customer'' or KYC~\cite{nel2017know,mondal2016transaction,gelb2019identifying}, which primarily lies with telecommunications companies. 
However, authorized third-party agents acting as last-mile service providers often assist in the KYC process. Agents often operate out of convenience stores or kiosks within the local community. In some jurisdictions, KYC processes are done only at the time of registration---a move partly designed to avoid excluding people who lack formal identification~\cite{gsmaKYC,gsmamomo2024}. 
However, this may increase the risk of fraudulent transactions~\cite{gsmaMomoRegulatory, gsmamomo2024}. Other countries, including Kenya and Tanzania, require KYC with each transaction~\cite{gsmaKYC}.


The role agents play in KYC, especially per-transaction KYC, introduces privacy concerns; agents may access sensitive information such as the size and recipients of transactions, user habits, social links, and other personally-identifiable information (PII). Prior work shows that users are concerned about potentially malicious agents, including illegal use of customer data to register additional SIM cards~\cite{sowon2023role, thebftonlineSafeMomo}. 

In pursuit of privacy, some users adopt practices such as using SIM cards registered in another person's name~\cite{luhanga2023user}, creating challenges for KYC compliance. Privacy-preserving techniques such as verifiable credentials~\cite{lacity2022verifiable,mazzocca2024survey}, homomorphic encryption \cite{sasson2014zerocash} and Zero-Knowledge Proofs (ZKPs)~\cite{bunz2018bulletproofs,miers2013zerocoin} are impractical for less sophisticated devices, such as basic phones, that are widely used for MoMo services. 
Moreover,  MoMo users often want to conduct transactions remotely, in which case they often send friends or family to agents. This presents challenges with both KYC correctness and privacy.

\textbf{The goal of this study is to design and test alternative privacy-preserving transaction KYC protocols that are practical for MoMo contexts.} We focus on MoMo protocols used in Kenya, which already use a per-transaction KYC process and where there are over 77 million registered MoMo accounts~\cite{statistaKenyaRegistered}.  
Thus, we designed protocols allowing MoMo users to deposit (cash-in) and withdraw (cash-out), while minimizing the exposure of users' personal data to agents. We also include a protocol for KYC in delegated transactions.

Our new protocols redirect sensitive data flows from agents---who conduct KYC checks---to the MoMo provider; our findings 
suggest that users trust MoMo providers with sensitive data more than agents. 
Users of the new protocol authenticate themselves to a digital ID service using biometrics, either on the user's own smartphone or using already-deployed voice authentication tools for basic phones~\cite{Jitambul31:online}. Users then share the authentication certificate with the MoMo provider, which provides a one-time code to the agent and the user. This code is used to confirm identity and allows the agent to transfer or collect cash from the user. We also extend our protocols to delegated withdrawal, which is a common use case (e.g., a user sends their friend to withdraw cash at an agent). The main design challenges were to make the protocols usable and compatible with resource-constrained devices (i.e., basic phones), and efficient (i.e., minimizing communication costs).

After designing these alternative KYC protocols, we investigate the following research questions (RQs):
\begin{itemize}
\item \textbf{RQ1:} What are users' privacy perceptions and what factors influence their data sharing attitudes in the context of agent-facing mobile money transactions?
\item \textbf{RQ2:} What security, privacy, and usability factors influence user and agent preferences and concerns for our alternative, privacy-preserving protocols? 
\item \textbf{RQ3:} What other design considerations arise with the use of privacy-preserving protocols? 
\end{itemize}

\figureCashOut

We present the results of a role-play and semi-structured interview study conducted 
with 32 MoMo users and 15 agents in Kenya to test our privacy-preserving protocols. Expanding on prior work~\cite{sowon2023role}, we surface  concerns that users have when sharing  data with agents, including fraud and concerns for personal safety. We find that both users and agents prefer our proposed protocols because they address inconveniences with existing processes and offer better security and privacy. However, participants also highlight issues that need to be addressed before these protocols can be deployed.

\section{Background and Related Work}\label{sec:related}

We introduce the current MoMo processes, their privacy challenges, and our proposed solution. Then, we discuss related work on the security and privacy of MoMo and digital lending apps. 
 While previous studies~\cite{bowers2017regulators,bowers-emerging-digital,sowon2023role,munyendo-loans-2022,akgul2024decade} highlight security and privacy issues with 
 financial apps in the developing world, none study the privacy of 
 KYC practices. 
We design privacy-preserving KYC processes for MoMo and evaluate their usability with agents and users in Kenya.

\subsection{Current MoMo Processes}
MoMo is widely adopted by both banked and unbanked users, with more than 640 million users worldwide in 2023~\cite{Mobilemo67:online}. Users include both smartphone and basic phone owners, with a significant portion of the population using basic phones. In Kenya, the basic phone market share exceeds 40\%~\cite{malephane2022digital}.

MoMo transactions include deposits, withdrawals, person-to-person transfers~\cite{sowon2023role}, and person-to-merchant payments e.g., via the ``Pochi La Biashara'' service offered by Safaricom in Kenya~\cite{safaricom-pochi-2021}. Withdrawals can be done in multiple ways (e.g., at an agent, or bank to MoMo transfers), but our focus is on the processes involving an agent. In the current withdrawal process in Kenya (Figure \ref{fig:curr_cashout}), a user presents their ID to an agent 
before initiating the transaction on their phone. Typically, they will open the MoMo app, select the transaction type, enter the amount and the agent's number, and confirm the transaction by entering their MoMo PIN. If successful, two confirmation messages will be generated (Figure \ref{figUIs}), one for the agent and one for the user. Once the agent confirms the transaction, they record user details (e.g., ID number and name) and transaction details (e.g., amount, time, and transaction code) before handing the user the money. 


\begin{figure}[t]
  \centering
  \includegraphics[width=0.8\linewidth]{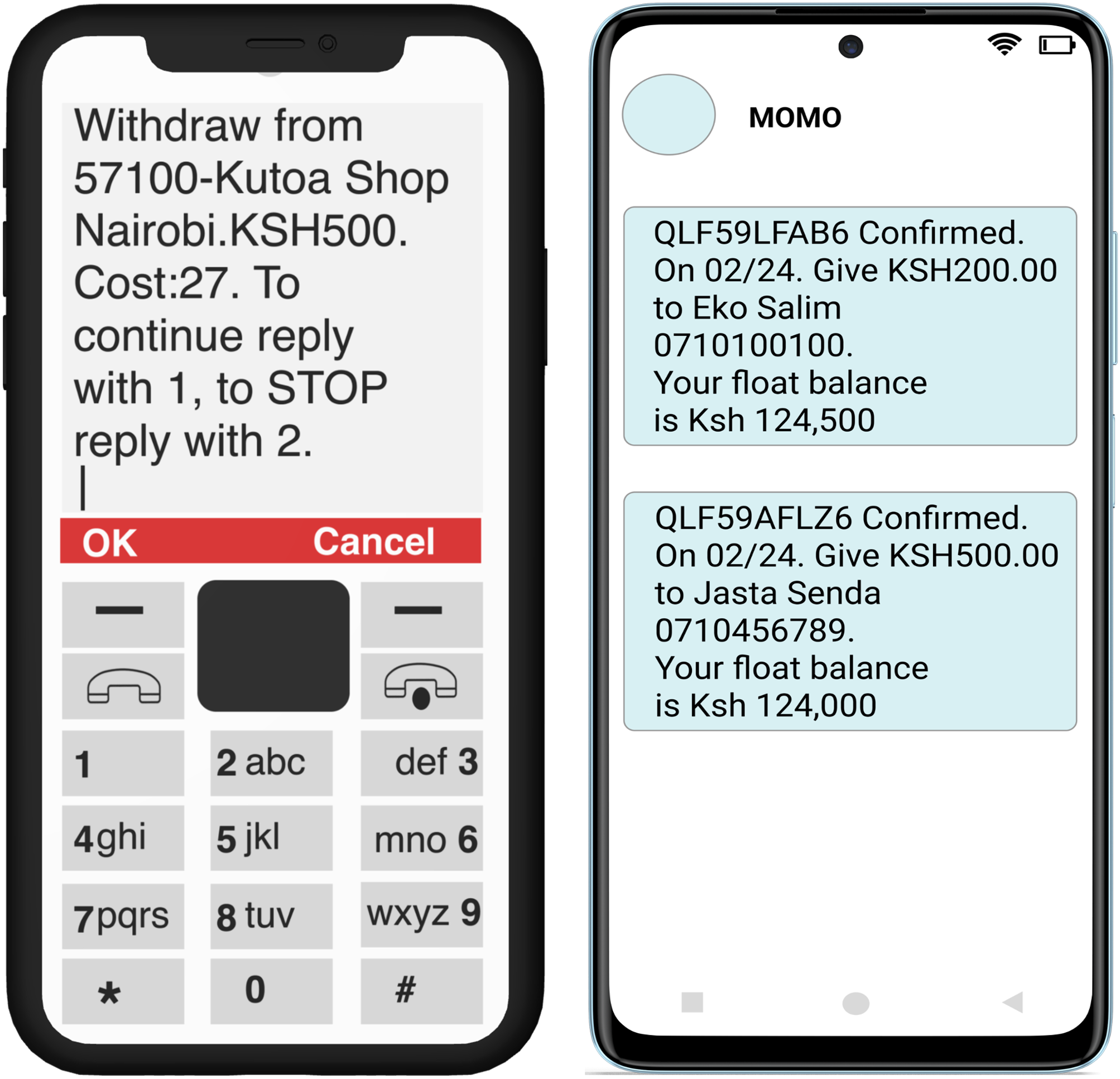}
  \caption{Sample MoMo UIs: Withdrawal acknowledgment on a basic phone (left) and messages for agents (right)}
  \label{figUIs}
\end{figure}

\mypara{Privacy challenges with current process} Previous work~\cite{sowon2023role} and our study show that many users worry that their PII recorded by agents can be misused under the current processes. 
We highlight some steps in the current process (Figure \ref{fig:curr_cashout}) that leak personal information, and what the risks are. In the KYC step, users submit their ID card to the agent to identify themselves. Agents are required to keep a record of all transactions that they facilitate in a physical book, including the ID number and name of each user. In addition to the ID card exposing many other pieces of data not needed for KYC, e.g. place of birth, cases have been reported of agents using users' IDs for illegal activities such as registering additional SIM cards---also mentioned by users in this study (see Section~\ref{sec:privacyperceptions}). Once the transaction is authorized, the agent receives a notification (Figure \ref{figUIs}) that contains more information such as the user's name and phone number. These details can be misused by agents for unsolicited phone calls or unauthorized data-sharing. 
Even when agents do not misuse information, they may not be able to store it securely. 

\mypara{Our solution} To protect users' personal information when withdrawing money, we designed an alternative, privacy-preserving process (see Figure \ref{fig:priv_cashout}). For this process, users do not need to hand their IDs to the agents nor share their personal information e.g., phone numbers with the agents. Instead, they authenticate using biometrics on their phones (fingerprint or facial for smartphones, and voice identification for basic phones). Afterward, users can initiate the withdrawal on their phones, generating a confirmation message with a unique code sent to both the user and the agent. Once the agent confirms the user received the same code, they can hand the money to the user. See Section \ref{sec:protocol} for more details.



\subsection{Security and Privacy of Mobile Money}
While MoMo services 
offer financial services to otherwise unbanked users in the developing world~\cite{donovan2012mobile,jack2011mobile}, there is limited work exploring the security and privacy of these services. 
Reaves et al.~\cite{reaves2017mo} uncovered 
vulnerabilities in 46 Android MoMo apps used in 28 countries, with many containing botched certificate validation and other forms of information leakage, leaving users vulnerable to impersonation and financial fraud. A related study of Android MoMo apps similarly found that most of these apps were not following security best practices~\cite{darvish2018security}. Through a systematic analysis of 197 digital financial services in Africa and South America followed by interviews with stakeholders, 
Castle et al.~\cite{castle2016let} found that although these apps were susceptible to various attack vectors, service providers were making efforts to secure them. 

Bowers et al.~\cite{bowers2017regulators} investigated 54 mobile money services in 32 countries in 2017, finding that almost half did not have any form of privacy policy. For those that did, these policies were hard to read, and often not in the language of their target users. This makes it difficult for users to understand what data is collected from them, how it is used, and how it is secured.

MoMo services use  human agents to facilitate transactions. 
However, the interaction between users and MoMo agents introduces security and privacy challenges~\cite{mogaji2022dark}. Through 72 semi-structured interviews in Kenya and Tanzania, Sowon et al.~\cite{sowon2023role} found that both users and agents design workarounds to the challenges posed by MoMo systems, including relying on their relationships for informal authentication. 
The study highlights the need to rethink the privacy and security of this ecosystem, along with usability. Accordingly, we design new privacy-preserving workflows for MoMo, and evaluate their efficacy and feasibility via role play and semi-structured interviews with both agents and users in Kenya.

\subsection{Emerging Digital Lending Apps}
Digital lending apps have emerged as a quick and easy way to obtain loans in the developing world, enhancing financial inclusion~\cite{faux_2020,hecht_2019,adams_2016}. This has been facilitated by rising smartphone usage and by MoMo services, which directly disburse loans to users' devices. However, these apps have also raised privacy concerns. 
Bowers et al.~\cite{bowers-emerging-digital} analyzed the privacy policies of 51 digital credit lenders and found they often failed to disclose data gathered by the apps. 
Munyendo et al.~\cite{munyendo-loans-2022} interviewed users of mobile loan apps in Kenya and learned about issues such as social shaming when users default in repayment. Similar concerns have been noted in India, sometimes driving loan app users to suicide~\cite{ramesh2022platform,saritha2023demystifying,ali2023fintech,aggarwal2024predatory}. Akgul et al.~\cite{akgul2024decade} analyzed reviews on the Google Play Store and found that these privacy concerns are widespread across many countries.  

\section{Protocol Design} \label{sec:protocol}
The objective of this work is to design a privacy-preserving KYC protocol for MoMo. 
The key design considerations were: 1) \emph{Correctness:} the protocol should correctly implement KYC for MoMo. 2) \emph{Privacy:} it should provide better privacy for users relative to the current process. We define privacy in terms of \emph{data minimization}: that is, minimizing the transfer of sensitive information to the agent. 3) \emph{Cost:} the protocol should minimize costs for users, including communication round trips to/from MoMo and transaction fees.
While some data minimization steps (e.g., not sharing PII with agents) may pose risks to agents who rely on user phone numbers to informally handle transaction issues (Section \ref{sec:security concerns}), this is actually a workaround in the current process. Transaction problems should be handled by the MoMo provider. In general, we find that our protocols close loopholes from workarounds and reduce agent liability, as data flows occur directly between the MoMo provider and the user.

Our design makes assumptions about trust in MoMo providers and availability of identity infrastructure. 

\mypara{Trust in MoMo providers} 
We assume that users trust the MoMo provider \emph{more} than individual agents, an assumption validated by some users in our study (see Section \ref{sec:privacyperceptions}).

\mypara{Availability of identity infrastructure} 
In line with prominent identity frameworks \cite{NIST2024digital}, we assume the existence of infrastructure for identity-proofing individuals by providing an ID card or other accepted documents for initial validation. 
We assume this infrastructure also enables digital authentication of users (e.g., using biometrics), which we believe is realistic given the efforts to pilot digital ID systems in Africa~\cite{DigitalIDAfrica}. Given the wide use of basic phones by MoMo users, we assume authentication infrastructure for both smartphone and basic phone users. For smartphone users, existing options include fingerprint or facial image recognition, while feature phone users can use voice biometrics. 
While photo ID authentication is still in its early days worldwide (e.g.,~\cite{singpass}),
Kenya (i.e., Safaricom) has already implemented a voice biometric system called \textit{Jitambulishe}, which allows users to identify themselves when they need to reset their PINs~\cite{Jitambul31:online}. 
The identity infrastructure could either be managed by the MoMo provider, 
a national digital ID service, or an approved third party. Our protocols are therefore robust enough to be used with solutions such as Verifiable Credentials (VCs) without any additional overhead or workflow changes. 

The protocols for smartphone users and basic phone users are exactly the same, with the only difference being the type of biometric used. In the following sub-sections, we illustrate only two privacy-preserving protocols for withdrawing: 1) withdrawing for basic phone users and 2) delegated withdrawing. The current (see Figure \ref{fig:curr_dep}) and corresponding privacy-preserving deposit (see Figure \ref{fig:priv_dep}) protocols are included on Figure \ref{fig:deposit} in the Appendix. 

\subsection{Privacy-Preserving Withdrawal}
Our privacy-preserving withdrawal protocol for basic phones is illustrated in Figure \ref{fig:priv_cashout}. 
The parts of the protocol that \emph{differ} from the current protocol (Figure \ref{fig:curr_cashout}) are indicated in blue.
In our protocol, a user initiates a transaction by supplying their biometric. 
Basic phone users will initialize a transaction and receive a phone call to authenticate with voice.
After successful authentication, the user proceeds by entering transaction details, such as the amount and the agent number for the store where they are transacting. The MoMo provider verifies that the agent number is valid and provides a transaction summary for the user to confirm on their phone by entering their PIN. Both the agent and user receive a confirmation SMS. Unlike the current process, where the agent SMS contains the user's name and phone number, the privacy-preserving protocol provides an $\ell$-bit code $c$ to both the customer and agent, which is valid for time $\Delta$ after it is generated. The agent also receives an amount associated with the code. After the agent verifies that their code and amount match those on the user's phone, the agent dispenses the specified amount of cash.

\mypara{Evaluation of Design Objectives}
(1) \emph{Correctness:} The new protocol may have KYC correctness similar to or better to the current protocol, as users are automatically authenticated by the MoMo provider, rather than by an agent. 
This can prevent fraudulent KYC practices discussed in prior work~\cite{sowon2023role}. 
While automated biometric verification can lead to errors, particularly in non-white populations~\cite{drozdowski2020demographic}, we note that Jitambulishe has already been used successfully in Kenya~\cite{Jitambul31:online}.
(2) \emph{Privacy:} The new protocol has better privacy protections against agents, who never see the user's name or phone number. However, they do see the transaction amount; we view this as insurmountable, as the client must receive or give cash to the agent. 
(3) \emph{Cost:} The protocol incurs 1.5 extra round-trips of communication from the client to the MoMo provider. 

\mypara{Summary of Privacy and Security Analysis}

We provide a basic privacy and security analysis in Appendix \ref{app:analysis}. At a high level, the core security argument is that users can deviate from the protocol by modifying \emph{who} is withdrawing money, \emph{how much} money is withdrawn, or \emph{when} money is withdrawn. We outline various cases leading to each failure mode, and show that the protocol fails with negligible probability in each case. 
For the privacy analysis, we note that our protocol is not cryptographically private, as the agent physically sees the user when they withdraw. However, we show that the protocol's privacy leakage to the agent is limited to the mutual information between the side information observed by the agent and the identity of the user withdrawing funds. Crucially, under our protocol, this side information does \emph{not} include the user's ID or phone number.

\figureCurrentProxy


\subsection{Privacy-Preserving Delegated Withdrawal}
Prior work on user-agent interactions in Kenya and Tanzania~\cite{sowon2023role} revealed that MoMo users often involve other people when withdrawing money. 
In these \emph{delegated transactions}, a \emph{sender} (the user who wants to withdraw) sends a \emph{collector} or \emph{proxy} (e.g., a friend or relative) to withdraw cash. 
There are different ways of executing delegated transactions using existing MoMo systems. In one variant, users send a peer-to-peer transaction to the proxy, who withdraws and collects cash on the sender's behalf. This incurs an additional cost.
The sender may instead share their ID, physical phone, and PIN with the collector. The sender may also withdraw remotely using saved agent details, and send the proxy to receive the cash. 
In this scenario, agents often operate on trust; they release cash to the proxy without verifying that the sender authorized the proxy to withdraw on their behalf (see Figure \ref{current_proxy}). 

Our proposed delegated privacy-preserving withdrawal protocol (Figure \ref{fig:delegated_withdraw}) formalizes the delegated cash withdrawal process by allowing a user to initiate a transaction and assign it to a collector who will travel to an agent store to complete the process. 
Unlike the current withdrawal practice, our proposed protocol ensures that both the sender and collector are authenticated to MoMo, thus facilitating proper KYC. 

Like the base privacy-preserving withdrawal protocol (Figure \ref{fig:priv_cashout}), the first step requires user verification through biometrics (Figure \ref{fig:sender_side}). The sender then assigns the transaction to a proxy by providing their phone number. When the proxy arrives at an agent store (Figure \ref{fig:collector_side}), they initiate collection by providing the agent number, and selecting the transaction that they want to collect. After confirming the transaction by entering their own PIN, the unique codes are sent to both the collector and the agent. The agent ensures that they are giving money to the right person if both codes match. At this point, the sender receives a notification about the completed transaction; only the sender incurs transaction charges. 

\mypara{Evaluation of Design Objectives}
(1) \emph{Correctness:} The formalized delegated protocol removes guesswork and inconvenience involved in authenticating users who transact remotely. 
With both the sender and collector authenticated to MoMo, the agent has better assurance of KYC for all parties. (2) \emph{Privacy:} The protocol offers more privacy guarantees to the sender, because they do not have to share any personal information with the collector or the agent. (3) \emph{Cost:} As no formal process for delegated transactions exists,  we cannot compare the communication round trips. However, the new protocol may reduce surcharges. In the informal proxy process, users incur double charges to withdraw if they send money to the collector's wallet so that they can withdraw it.  
We extend our security and privacy analysis to this protocol in Appendix \ref{app:analysis}.



\section{Methods}
To evaluate our proposed privacy-preserving protocols, we conducted a user study with MoMo users and agents in Kenya in March 2024. 
We conducted an in-person within-subjects usability study with 32 MoMo users where we adopted a wizard-of-oz approach, with one researcher simulating the responses of the phone based on the user's interaction with our low-fidelity paper prototypes. We compared participants' privacy, security, and overall design perceptions between the comparable steps in the current workflows and our proposed ones. The comparison was qualitative rather than quantitative.
%
We also conducted an interview study with 15 MoMo agents using a process walk-through demonstrating the protocols, as opposed to a role-play format. We did this because agents are typically small business owners, and it is difficult to get them to dedicate much time to the study.    
Each demo was followed by interview questions to gather agents' perceptions of the processes. 
All interviews were conducted in a mix of English and Kiswahili, as is commonly spoken in Kenya. The interviews were conducted by a native researcher assisted by two native researchers. The data set is available at https://doi.org/10.1184/R1/29176322. 


\subsection {Study Procedure and Data Collection}
\mypara{Prototypes} We created the privacy-preserving protocols as low-fidelity prototypes using Figma and printed the finalized prototypes for the user experiments. To control for learnability, we designed our prototypes to maintain the look and feel of Safaricom's M-Pesa~\cite{jack2011mobile}, the most popular and widely used MoMo product in Kenya~\cite{MPesa–Af52:online}. 

We designed four interface variants, with combinations based on phone type (smartphone or basic phone) and language (English or Kiswahili). We showed each user the variant and language corresponding to the phone they used.
The same wording of notifications and prompts was used across variants. Due to limitations of character counts on basic phone displays, we intentionally excluded some message elements (e.g., date and time), 
allowing us 
to use the same prototypes for the whole study. 
We also created paper prototypes of the current process 
and materials that users would see if a MoMo provider implemented the new protocols, such as an ad showcasing the biometric authentication features and a video demo of the biometric enrollment process. 

\mypara{Participant Recruitment} To recruit MoMo users, we posted flyers at agent shops in five diverse neighborhoods in Nairobi: Nairobi West, Kilimani, Highrise, Umoja One, and South B. We also visited agent shops in these locations to directly recruit participants. 
In addition, we asked participants and other people in our networks to share information about our study. We advertised the study as an interview about alternative MoMo processes, without mentioning privacy, to minimize social desirability bias and priming. 
We also tried to recruit the agents who owned the shops for the agent part of our study. We aimed for a balance between male and female agents, 
as past research has shown that gender influences how people use MoMo services~\cite{chamboko2020role}. We succeeded in recruiting eight male agents and seven female agents. Table \ref{table:demographics} in Appendix~\ref{app:additional} provides a summary of participant demographics. 

We shared study information and scheduled appointments over the phone or in person. Those who agreed to participate received reminders before the scheduled interviews. 
We emphasized that participation was voluntary and that participants could opt out without any consequences. We compensated all participants for their time and transportation costs at rates approved by the Kenyan ethics board (approximately \$5.40). 

\mypara{Interviews and Tasks} We introduced the study as a product test for a fictitious company to avoid participant bias. Before starting, we briefed participants on what to expect and demonstrated the concept of thinking aloud. We also tested each participant's ability to read by asking them to complete the following unrelated tasks using a list of menu items on a screen and asking what they would select to: save for their business, identify themselves, and make a phone call to a number not saved on their phone. The menu items included in this literacy test section aren't included in the current MoMo interfaces. One participant was excluded this way. 

Thereafter, participants completed three tasks: withdraw money using the standard protocol, withdraw using the privacy-preserving protocol, and complete a delegated withdrawal with the privacy-preserving protocol (as both a sender and a collector). 
 All participants completed these tasks in the same order. Each task was followed by interview questions to elicit participant feelings and perceptions about the protocol. 
 
 We asked participants to think aloud as they completed the tasks. We recorded their audio and videos of their interactions with the paper prototypes, capturing only video of their hands for confidentiality (Figure \ref{figSetUp}). After each task, participants were asked about their experiences and perceptions of the protocols. The experiments, along with the  post-task interviews, lasted approximately 60 minutes.
\begin{figure}[t]
  \centering  \includegraphics[width=0.8\linewidth]{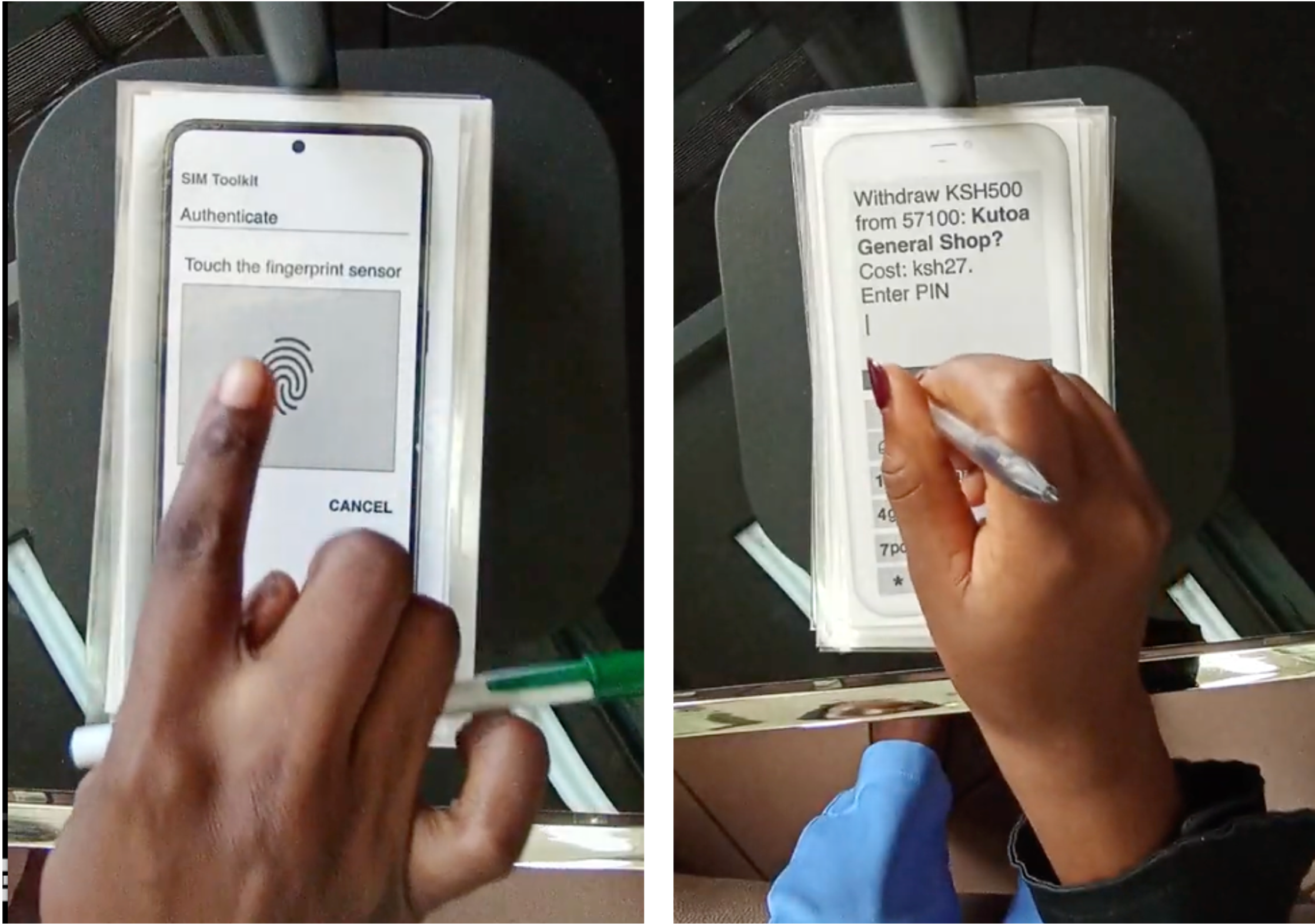}
  \caption{Users interacting with the paper prototypes}
  \label{figSetUp}
\end{figure}
The interviews with agents followed a similar flow with questions focusing on their roles as agents. These lasted 30 minutes on average. 
Both agent and customer interview scripts are in Appendices \ref{app:interview-agents} and \ref{app:interview-customers}.




 \subsection{Data Analysis} The audio-files were translated to English and transcribed before being uploaded to Nvivo v14 for coding. The coding was done by three researchers, two of whom are Kenyan natives. We used inductive analysis to develop an understanding starting from the data rather than from predetermined codes.
We had three major stages of coding. The first stage entailed development of the initial codebook where three researchers independently coded an agreed upon sample of the transcripts. They then met to discuss emerging codes, come to a shared understanding of their meanings, and develop the initial codebook. This was done iteratively until agreement was reached.

In the second stage, the coders used the initial codebook to inductively code the remaining transcripts. This was completed cross-sectionally by splitting the transcript into three sections based on the user tasks and assigning each coder a section across all the transcripts. 
The coding team met weekly to discuss the suitability of existing codes, add new codes to the codebook, and resolve any differences. 

The final stage was a peer review process, 
in which each coder double coded the whole cross section of the dataset for which  they were not the primary coder, and compared their codes with the primary coder's. This process of double coding helped identify inconsistencies, which were discussed and resolved between the primary and secondary coders.  
Since all transcripts were coded independently by two separate coders and all conflicts were discussed and resolved, there was no need to compute inter-rater reliability~\cite{mcdonald-reliability-2019}.

After all coding was completed, the primary coder reviewed the codes to combine, group, and categorize similar ideas before abstracting them to create themes and sub-themes based on the research questions 
 as summarized in Table \ref{tab:codebook} in the Appendix. The resulting themes were discussed with the wider research group and refined further. 


 \subsection{Ethics and Positionality}
\mypara{Ethical Considerations} This study was approved by two Institutional Review Boards (IRBs) in the U.S. (Carnegie Mellon University and George Washington University), and one in Kenya (Strathmore University). 
We also received a research permit from the National Council of Science and Technology (NACOSTI) in Kenya. We took several steps to protect participants. 
First, we assigned the participants a fictitious profile with a name, a phone number, and an ID to use throughout the study to avoid sharing their own PII.
Second, we sought consent to audio record and video record the hands of the participants. We planned to use the video recordings to offer additional insights into participant interaction with the prototypes. In the end, the think-aloud data sufficed. 
Although we collected participant phone numbers for follow-up and compensation, this information was stored separately and not linked with study data. Access to the audio files was limited to researchers and a transcriptionist. 


\mypara{Author Positionality}
Four of the researchers, including the lead researcher, identify as Kenyan natives and have extensively interacted with MoMo. 
This informed the design of the study and helped ensure correct interpretation of the data without losing contextual nuances. The non-Kenyan members of our team brought fresh perspectives that further enriched our protocol design and the analysis of the results.

\subsection{Study Limitations} Our study has some limitations. First, due to the qualitative nature of the study and the small sample size, we cannot make any generalizations. Second, the study may be subject to biases  such as participant self-censorship and recall bias. These may be pronounced given that participants were not engaged in actual transactions and did not have any personal information at stake. Third, our paper prototypes are likely more cumbersome to use than a digital interface, which may have impacted participant perceptions. Fourth, all participants were shown workflows in the same order, which could have introduced ordering effects. However, through our pilot testing, we found that first reviewing the existing protocol, and then proceeding to the most different proposed workflow was most effective for participants understanding the exercise. Finally, some participants were confused when role-playing both the sender and collector in the delegated protocol.   

\section{Findings}
\label{sec:Findings}
In this section, we present our results. 
Our analysis aimed to surface general themes about privacy in the context of MoMo transactions and participants' experiences with the new protocols. To avoid implying generalizability, we report our qualitative findings using the following: a few (less than 25\%), some (25-45\%), about half (45-55\%), most (55-75\%), and almost all (75-100\%). We identify user participants with the letter U and agent participants with the letter A. 

\subsection{Data Sharing Perceptions and Attitudes}
As in Sowon et al.~\cite{sowon2023role}, we found that users were concerned about their privacy in their interactions with agents. Our findings provide  richer insights regarding factors that contribute to users' data sharing attitudes with agents. 

\subsubsection{Privacy perceptions}
\label{sec:privacyperceptions}
We present the privacy insights that emerged as a four-dimensional typology capturing what is private, private from whom, risks, and risk mitigation techniques.

\mypara{What is private and why}
Most users felt that account balances and ID information were private, while almost all said their PIN is private. About half felt that their phone number was personal and very few considered their name private. To justify these perceptions, some cited the personal identifiability of the data, while others discussed how  their ID was a gateway to many pieces of information: ``If someone has my ID number, they can track my NSSF,\footnote{The National Social Security Fund allows employers and citizens to contribute to save for retirement similar to the 401K in the US.} NHIF,\footnote{The National Health Insurance Fund is a contribution-based fund to provide accessible and affordable health insurance for Kenyan citizens.} [and] my ID number links me to my bank and my family details'' (U18). 
A few users said they would consider their data as private under some conditions e.g., large transaction amount. ``[The amount is] personal depending on the money transacted. If I have withdrawn a million and someone knows, they might plan to attack and rob me'' (U22).

\mypara{Privacy from whom} Participants frequently expressed concerns about who had access to their information during transactions, with almost all users uncomfortable with agents having access to their data. In addition to sharing specific concerns (summarized in Table \ref {tab:user_concerns} in Appendix~\ref{app:additional}), a few explained that too much data was shared, and others expressed concerns about the recording of data  in a book for evidence of transactions. 
Only a few users were concerned about the MoMo providers having access to their information. The most common reason for concern by this group was a mistrust of MoMo employees who could access account balance data, potentially identifying those with large balances. 

When asked about sharing data with proxies in the current proxy engagement techniques, some users were worried about sharing personal details e.g., ID or PIN, even though most people send proxies that they trust.

\mypara{Risks that concern users} Almost all users discussed the negative outcomes that could arise from data sharing during  MoMo transactions. 
We find four main types of concerns: fraud and theft, unauthorized use of data for non-MoMo purposes, physical harm, and economic judgments. 
The most frequently mentioned harms were fraud and theft, followed by unauthorized use. While most cited the agent as the main threat actor, a few users feared that MoMo employees were involved. 
Users mentioned unauthorized uses of data by agents such as registering new SIM cards and registering voters. ``I came to learn that my ID number had registered another SIM card and I didn’t know'' (U18). Some who mentioned the risk of physical harm tied it to the agent knowing user account balances and potentially acting as accomplices to defraud them. Others mentioned stalking and harassment.``[The agent] might try calling me and stalking me, and so I wouldn’t want them to have a name and phone number to avoid that'' (U19). Finally, a few users shared that there might be a risk of socio-economic judgments due to knowing how much money one transacts or has in their account. 


\mypara{Risk mitigation techniques} 
Some users suggested potential privacy protections. 
For example, some suggested blurring or removal of parts of PII. This could include redacting portions of one's ID or phone number, which is similar to the practice of redacting all but the last four digits of the Social Security number in the United States. 
Users provided these suggestions to balance their desire for privacy with their desire to provide the data they perceived as necessary to complete transactions. 

Some users who shared data with proxies reported redacting the account balance for the shared notification message, changing their PIN when the transaction was complete, and auditing their account to ensure the proxy did not withdraw a larger amount. ``She went to the shop and came back with the money and after that I changed my PIN'' (U30).
 
\subsubsection{Factors influencing data-sharing} \label{subsec:datasharing} 
We explored participant beliefs about the data exchanged in MoMo transactions 
and find that perceived utility of shared data, perceptions of data sensitivity, trust, and obligatory compliance contribute to data-sharing acceptability.

\mypara{Perceived utility of shared data} Almost all users and agents mentioned the practical utility of sharing various pieces of information for authenticating and verifying users, monitoring transactions, and recourse. 
``I verify [who the customer is] by asking for his or her ID'' (A1).
Most agents and users mentioned scenarios where sharing data could help them follow up with each other or contact the provider in case of issues. Agents specifically relied on the data collected during the transaction process for recourse.
\begin{quote}
    [When] a person has withdrawn, he can call Safaricom, and claim he withdrew from a wrong agent and so you see when Safaricom calls me as an agent and they ask me “did you get the ID?” if you did not get the ID, Safaricom will automatically send the money back to the person. (A14)
\end{quote}

\mypara{Perceptions about data sensitivity} Many users were comfortable sharing data they felt was ``less sensitive'' for various reasons. Most commonly, users felt that data such as their name and amount were ``common knowledge'' to other people. ``Everyone calls you by your name, so it is not a secret'' (U18). A few participants were open about sharing their information because they ``had nothing to hide'' or because they believed they were anonymous to MoMo: ``I am not doing an illegal transaction like I have stolen anyone’s money, it’s very legal so I don’t have anything to hide'' (U28).

\mypara{Data sharing based on trust} A common theme was the requirement of trust when completing transactions. About half of users trust agents and providers, with MoMo providers being trusted slightly more than the agents. For instance, U22 said: ``They are the service providers offering the service. I trust them to keep my information secret.''


For delegated transactions, 
trust also played a central role in how users decided who to send, and to which agent.
Even for the privacy-preserving delegated process, where participants acknowledge the benefits of not having to share data with proxies, about half of users still cited trust as a requirement for sending someone else. Ultimately ``once the [proxy] has the money I [still] need them to actually give it to me'' (U3).

\mypara{Sharing out of obligatory compliance} Most users and agents expressed the need to share their data for compliance when completing MoMo transactions, or just as a necessary trade-off to use MoMo services: ``There are details they might need and even if they are personal the company says it’s okay to give them, I am going to give'' (U26).




Users seem to accept the requirement of sharing their information as part of the MoMo transaction process, but when shown alternative options, they understand the potential risks associated with those actions.

\subsection{Existing Process Inconveniences}

When asked if there was anything they would change about the current MoMo transaction processes, many users expressed general satisfaction, describing the process as simple. However, further questioning about their current practices revealed existing process inconveniences, which they considered normal.
The following four inconveniences were prevalent: ID-related inconveniences, balancing data-sharing and security, contextual complexity in security and privacy decision-making, and usability challenges. 

\mypara{ID-related challenges}
Many of the agents and users expressed dissatisfaction with the use of IDs in MoMo. 
Users noted usability challenges with IDs, whereas  agents felt that the use of ID was inefficient. The reasons given ranged from being cumbersome for both them and their clients, jeopardizing users' privacy, and not being fully secure. 
\begin{quote}
    [The use of IDs] is challenging and there should be a simpler way of [verification] to speed up the process \ldots [Because] if you have five people that are there on the queue waiting, the process of checking, and writing [IDs] is time consuming and maybe that person urgently needs the money.  (A10)
\end{quote}
One agent pointed out the inadequacy of IDs for security: ``We are in Kenya where people can go to [duplicate IDs] and have the same ID as someone else. These people that do fraud on M-Pesa nowadays have these fake IDs'' (A12). We note that agents' concerns are largely not fully addressed by our protocols, which still rely on the ID infrastructure.

\mypara{Inconvenience of balancing data-sharing and security}
When sending others to complete transactions, most users found themselves having to share a lot of their personal information with the proxy. Most people indicated giving the proxy either their phone and their MoMo PIN, or sharing other transaction details such as the 
notification message 
and other personal information, such as the physical ID card. ``I gave him my ID card and phone and gave him my PIN.'' (U14). U23 said: ``I will forward the message for them to show the agent the message'', while U18 explained that they did not like that they shared their data, but felt that it was their ``only option'' when they needed to send someone to collect money. 

\mypara{Contextual complexity in decision making} Both agents and users experienced contextual complexity from the current processes. For users, this mostly involved balancing factors such as the choice of agent, their location, and the transaction amount with privacy needs. On the agent's side, the complexity occurred in how they balanced the requirement for KYC with usability for their clients. 
Agents adopted more relaxed ID and transaction practices with known users or small transactions, while requiring stricter verification for unknown users, higher value transactions, or highly risky transactions like delegated withdrawals. ``[I require the ID] for new users [because] that is where there is a lot of fraud'' (A12). 

 
Other challenges manifested as process inefficiencies. For instance, when asked about how they currently manage delegated transactions, both groups shared various techniques e.g., 
ensuring that the collector bears the transaction and sender information, and offline communication between the agent and customer to hash out this informal transaction process. 
\begin{quote}
  When that person who has been sent comes I don’t give them cash first, I call that customer first and verify like `whom have you sent--in terms of their names and their number, and ID number'. After talking to him or her, then I can confirm that they have indeed sent that person. (A12)  
\end{quote}
Some agents even store user details for efficient remote and proxy transactions: ``Before you come to trust users that way, at least you have their details. All those users I have their ID’s'' (A14). 
Such processes introduce risks and inconsistencies, as agents constantly balance security with convenience. 


\subsection{Impressions of Proposed Protocols}
We find that participant preferences for the new protocols reflected perceived affordances beyond those offered by existing processes. In this section, we first highlight the specific security, privacy, and usability features that our participants identified as being advantageous in the new protocols. Next, we present usability and security issues associated with the new protocols that concerned our participants.
\subsubsection{Better privacy and security affordances}
Both groups of participants expressed preferences for privacy-preserving protocols, mentioning three categories of preferences---those related to security, privacy, and usability.

\mypara{Privacy-preserving process offers better security} Users and agents mentioned features that they felt provided better security. For delegated withdrawals, both groups felt that the verification of all parties, the transparency of the process, and the formalized way to use a proxy enhanced security.
\begin{quote}
   It is secure and procedural---the agents themselves have a way to trust it, and myself and the [collector] can trust it. Everything is documented, just in case we encounter issues, it can be traced. (U10)
\end{quote}

In addition to the delegated withdrawal features, almost all users and some agents felt that the code offered additional security. ``The code is another added layer of security. [The agent] and I are the only ones getting the code, and she asks me for the code and I tell her the same thing she has'' (U13).

Finally, the use of biometric authentication enhanced perceived security. ``I like it because I feel it is secure and someone will not be able to remove your money because they will not have your fingerprints'' (U12). 
Many agents shared similar feelings about biometrics. 
``People steal [other people's] ID's. [With] this one, [the agent is] sure [because] there is no way your fingerprint will match my fingerprint'' (A13).

\mypara{Privacy-preserving process offers better privacy} When asked what they liked or if they would use the new protocols, almost all users cited the benefit that their PII was not shared with the agent. Most also felt that the privacy-preserving protocols gave them better control over their data. 
\begin{quote}
 [I like it] that I just withdraw, and [the collector] collects for me so there is no personal information shared. I remain with my phone and my details and he uses his phone and his details \ldots there is also no way [the agent] can use my details [unlike] the first one [where they] can get my phone number and start calling and harassing me. (U9)   
\end{quote}

Agents also liked the privacy offered by the new process. A10, referring to the physical book agents maintain for transactions they complete, said: ``We have manual books and every time they are insisting we put down the [customer and transaction] details in case of anything---but this one is fully digital, and protects the privacy for the clients.'' 

\mypara{Privacy-preserving process is more usable} About half of  users and some agents specifically mentioned how the delegated withdrawal process offers a more convenient way to complete delegated transactions. 
Most users found the protocols relatively ``straightforward'' and ``easy'' to use and indicated that they would choose the new privacy-preserving processes over the existing ones. Some mentioned that there was more flexibility in choosing a proxy, mainly because the sender does not have to give someone their physical device, while others mentioned more freedom from not having to rely on known agents only, as they previously did with the delegated transactions. A few also felt that the privacy-preserving proxy process would be more affordable, since there would be no additional charges for collecting money for another person. Others appreciated the new proxy system for its increased geographic flexibility when compared to the current protocols where remote transactions are discouraged by the provider based on the distance between the customer and the agent. 

One recurring aspect that participants frequently discussed was the convenience offered by the process not requiring them to carry their IDs to transact. ``It is good because there is no need for an ID, you just use your fingerprint to transact [and] your phone [which] I always have it with me, so there is no burden. We keep forgetting to carry our ID’s so this would really help.''(U14). Two users also thought that the new authentication process would be more convenient for agents. 

Almost all agents felt the new process was more convenient for their users to authenticate. Two agents pointed out how better and efficient KYC impacted their own business success as agents. The process addresses ``the challenge of users coming to transact and they don’t have their identification cards and so you cannot allow them to transact but with this, it will be good for our business as well'' (U3). 

\label{sec:protocol challenges}

\subsubsection{Usability and privacy concerns}
Only a few users and agents expressed usability and privacy concerns, which fall in four main categories: process inconveniences, process complexity, privacy concerns related to voice biometrics, and access and accessibility concerns. 

\mypara{Inconvenient process} 
Agents and users mentioned potential inconveniences of authenticating each time with the new processes. A6 preferred the flexibility offered by IDs in some cases: ``You realize there are those users who frequently come here and you don’t need to keep authenticating their identity.'' 

Both groups also shared concerns with the new processes having many steps. Specifically, the voice biometric process requires users to dial a short code, listen to a prompt, and repeat a passphrase: ``it would be better if we had an option where you press and speak in without listening to that [prompt] \dots --- that would shorten the process.'' (U31).

 Given a choice between the existing and new processes, however, most participants who expressed dissatisfaction with the authentication process because of its length still chose the privacy-preserving process despite this inconvenience.

 A few agents and some users were concerned about the need for the proxy to have a phone or be registered. ``I cannot send my kid. Maybe he or she is a teenager without a phone.'' (U12). Other users worried this process would be restricted to literate users, since prompts are given as text on a series of screens, while some desired a batch collection option to collect money for multiple people at once.

\mypara{Complex protocols} Concerns about protocol complexity were mostly linked to a misunderstanding about the code, its expiration, whose details to enter  (sender's vs proxy's), and the two-step process completed by the sender and their proxy. 
Agents were concerned about the potential malicious use of an expired code by another user to get money from an agent. This is based on an incorrect understanding about how the code works. More importantly, both agents and users raised concerns with the code expiration window in light of potential network delays and 
and felt that these need to be ``shared very fast'' so the agent does not ``
have to keep users waiting because you are waiting for the code'' (A14).

Process confusion sentiments were mostly related to the delegated withdrawal process. Some agents described the process as ``complicated'' and ``confusing.'' For users, the confusion manifested 
during the delegated withdrawal task process 
when they had to play both the role of the sender and collector as they sometimes forgot which one they were assuming. 
\begin{quote}
I was confused---but I remembered I was [supposed to be] withdrawing, not collecting. [As] with anything new, things are a bit confusing, but with time we get used to it and it becomes easier. (U2)
\end{quote}
\mypara{Privacy concerns with voice biometrics}
One participant perceived voice biometrics as inconvenient and less private because they would have to ``shout'' every time they went to an agent, making them prefer an ID: ``I will use my ID, because you don’t have to shout around people about your password. It's quieter than this one'' (U25). 
Note that the Jitambulishe voice biometric does not require users to state sensitive information like a PIN---rather, they are asked to repeat a standard phrase, and their voice is used to authenticate. 


\mypara{Access and accessibility issues} Most agents were mostly concerned about the potential challenges with access to the privacy-preserving features as well as accessibility of the biometrics. The access issues they mentioned were mostly related to technology and demographics. These agents felt that issues such as malfunctioning phones and phone battery challenges would limit access and that older and illiterate people would struggle with the new protocols. Agents also felt that people with physical deformities would be disadvantaged and potentially excluded, especially with the use of fingerprints for verification. ``Maybe [the person] is a construction worker and their fingerprints are not quite clear and so sometimes they might try to use the fingerprint and it fails'' (A5). 

\subsubsection{Security concerns}
\label{sec:security concerns}
Agents and users alike expressed concerns about potential security issues with the new protocols. 
About half of the users believed that 
agents would need to receive additional information, including their name or phone number, to help with verification or potential recourse. About half of the agents shared similar concerns, wanting to see a customer's details in case they needed to reach out. A few agents were worried that the code for the proposed processes was not adequate to authenticate users, preferring to see the customer name. 

Another security challenge shared by both groups was the potential failure of biometric systems. 
For example, for voice biometrics, U1 asked: ``What if the voice changes when you get sick? Or when your voice is hoarse.'' 
Similarly, a few agents shared concerns about changes in one's voice. A few agents and users also felt that this system could be exploited. A13 discussed a potential distrust of users verifying themselves: ``\ldots for my security purposes, how do you really trust the customer to allow the customer to do self-service?''

\section{Discussion}
Current MoMo platforms are fraught with privacy, security, and usability challenges. These are exacerbated by reliance on agents and inconvenient, or exclusionary, KYC practices. Existing processes expose users to significant risks of unauthorized data use, while relying on traditional ID-based KYC creates usability barriers for underserved or marginalized populations. The growing reliance on MoMo services in developing economies necessitates the development of privacy-preserving protocols considering both technical constraints and human factors. To our knowledge, our study is the first to design and test privacy-preserving protocols for MoMo accounting for both external and internal threats, including risks posed by MoMo agents. Here, we discuss some takeaways. 

\subsection{Balancing privacy, usability and security}
Privacy remains an essential question in many digital financial systems, and MoMo 
is no exception---with studies highlighting various privacy issues
~\cite{reaves2017mo, bowers-emerging-digital, bowers2017regulators, munyendo-loans-2022, sowon2023role}
. 
Although privacy perceptions often influence data sharing~\cite{dinev2013information, malhotra2004internet, xu2011information}, 
we also know that system design contributes to better or poorer privacy experiences for users, which leads to our first takeaway:

\mypara{Takeaway 1: MoMo users are uncomfortable sharing KYC data with agents, but do so out of necessity} Both users 
 and agents view data sharing as a necessary part of transaction protocols, but users would prefer to share less, and agents would prefer to collect less. 
 Most user concerns about sharing data---such as fears of physical or psychological harm---have the agent as the main threat actor. 
However, systematized data collection 
processes, such as mandatory SIM registration, remain in place to comply with KYC regulations.

\mypara{Takeaway 2: Both users and agents prefer the privacy-preserving protocols for both their privacy capabilities and increased transaction security} 
Data minimization does not have to sacrifice protocol requirements. In this work, we proposed and tested one method that minimizes process-driven data-sharing, and supports users' preferences for not sharing more than required while ensuring 
reasonable efficiency and transparency.
A few participants indicated that they would have preferred more PII in the transaction notification because of its perceived utility, e.g., to provide assurance that the agent was transacting with the correct person. Even then, most of these participants felt that this information should be redacted from the agent's view. Overall, users preferred the privacy-preserving protocols because they do not share their personal information with the agent, and also provide a mechanism to securely engage the proxy without having to share PINs and phone. 
Similarly, agents expressed a preference for the new processes because they protect users' data and because they offer more security and transparency in delegated transactions, thus mitigating some of the associated risks.

\subsection{Digital identity verification and eKYC}
Digital identity verification is seeing growing adoption~\cite{singpass}. 
In digital financial services, identity verification entails establishing ``to some degree
of certainty, a relationship between a subject accessing online services and a real-life person'' to comply with KYC~\cite{NIST2024digital}. We find that MoMo users and agents are open to using biometrics for verification.

\mypara{Takeaway 3: Authentication using biometrics is feasible and preferred to ID-based verification} The use of IDs for KYC was a common inconvenience cited by both agents and users. As a result, participants almost always preferred the use of biometrics over IDs. Previous work shows that the lack of legal identification is one of the biggest barriers to financial inclusion~\cite{yongomobileidentity, martin2021exclusion}, and a critical barrier to KYC~\cite{gsmaKYC,gsmamomo2024}. While this affects MoMo users directly, it also indirectly impacts agents who benefit from the economic opportunities that MoMo creates, as one agent said: ``I prefer the new process [because] it addresses the challenge of customers coming to transact and they do not have their identification cards and so you cannot allow them to---but with this, it will also be good for our business.'' The inconveniences of IDs that users face also inconvenience the agent and present additional risks when they adopt workarounds. 
As more service providers use new digital eKYC processes to streamline KYC, our study provides a proof of concept with empirical results supporting the use of biometrics in a resource-constrained context. 

\subsection{Other Opportunities and Challenges}
While our protocols show promise in improving MoMo privacy without sacrificing security and privacy, they have some limitations. This leads to our fourth takeaway:

\mypara{Takeaway 4: New protocols are needed to accommodate delegated transactions where the proxy does not have a MoMo account or phone} 
Our findings show that participants often send people in their social circles to collect money. With the formalized delegated withdrawal, the proxy must have their own phone and be a registered MoMo user. This would be problematic in situations where phones are shared among family members or when the proxy is a minor and therefore not a registered user. Though these may not be an issue in places like Kenya where phone penetration is high~\cite{KenyaPhoneStats:online}, we believe that the impact of such a limitation on financial inclusion needs to be further understood.

\mypara{Takeaway 5: Data minimization in the new processes may present some vulnerabilities to agents} 
Some agents felt that privacy-preserving workflows left them with inadequate information for critical processes, such as recourse. Agents indicated that they might need users' personal details, such as name and phone number, to reach out in case of an issue. These scenarios included dispensing the wrong amount, especially excess cash. Those who preferred ID-based KYC indicated that it served as legal evidence in case of fraud. We thus note that agents are concerned about potential fraud they may face from malicious clients.    

\section{Recommendations}
While our protocols were largely acceptable to users and agents, there remain areas of improvement. However, this qualitative study may not be generalizable to all populations. To improve our protocols, we present five recommendations:

\mypara{Recommendation 1: Pilot the proposed protocols and test the placement of authentication in the workflow, to understand when, and why, users may not consider authentication necessary; then design alternative solutions that balance the goals of privacy, security, and usability}
A pilot of these protocols by MoMo providers would help to confirm their feasibility. We also note that, in practice, users may find the new authentication cumbersome if they have to complete it every time they use MoMo---even when completing transactions such as purchasing airtime. This experience appears to be consistent with the required use of IDs for transactions. In investigating authentication inconveniences, it would be beneficial to determine if this applies to all three types of biometric authentication or just to voice biometrics, as we found the most concern with voice authentication.  

\mypara{Recommendation 2: Study the extent and nature of the fraud targeted at agents, and ensure that MoMo workflows are designed with appropriate and adequate security protections from malicious users}
Some of the agents we interviewed reported security concerns and fraud targeting agents. 
While our privacy-preserving protocols were designed to have all the necessary protections to prevent fraud from ill-intentioned users and ensure that if fraud happens the agent has recourse, some agents  said they would still prefer to have client contact information for the sake of recourse rather than relying on the Momo provider's formal dispute-resolution process. Further work is needed to better understand agents' concerns and make sure they are adequately addressed. 

\mypara{Recommendation 3: Study the impact of limitations requiring proxies in delegated transactions to be registered users or have their own phones} It is important to continue pursuing inclusion in the context of MoMo. Previous research~\cite{luhanga2023user} has found that challenges with ID acquisition may negatively impact the registration of a SIM card in one's name. Such barriers can negatively affect the successful implementation of the proposed protocols, as users may find them more inconvenient if they cannot send the people they would typically send. Understanding the extent of such a limitation in the proposed protocols will help to make appropriate decisions about whether or how to resolve such a limitation.

\mypara{Recommendation 4: Study how to make KYC systems both privacy-and-security-preserving by design} Regulators and central bankers have consistently argued that achieving a fully private payment system would be incompatible with anti-money laundering and countering the financing of terrorism  requirements~\cite{Armelius21onthepos, Auer21CBDCquest}. This tension is especially evident in KYC where security is prioritized over privacy. This perspective potentially stifles innovations that could enhance user privacy while still meeting regulatory needs. A more nuanced approach would involve re-evaluating these assumptions to explore privacy solutions that minimize data exposure while still fulfilling identity proofing requirements. 



\mypara{Recommendation 5: Conduct large-scale surveys to determine the prevalence of privacy violations}
Our study was qualitative in nature and therefore limited in sample size, which is inadequate for statistical generalizations. Though we provide in-depth contextual insights into privacy and usability issues, topics currently understudied in the context of MoMo, a larger quantitative survey would be useful to assess the prevalence of privacy violations in countries that use MoMo.

\section*{Acknowledgments}
We would like to thank Peter Mwangi, Vema Oluoch, Lydia Kamuyu, Esther Adwets, Minnie Ogachi, and Cynthia Chepkoech for their assistance with recruitment and other study logistics.
This study was supported in part by the Gates Foundation and NSF grant SaTC-2325477. The views and opinions expressed, however, are those of the authors and do not necessarily reflect the views or positions of the sponsors. 

\bibliographystyle{plain}

\appendix

\newcommand{\mychoice}[1]{{$\circ$}~#1 \, }
\newcommand{\questionspace}{\vspace{0em}}

\section{Interview Script - Agents} \label{app:interview-agents}
\footnotesize

\noindent Thank you for participating in our study. We are in the process of testing mobile money transfer processes for a new product called MoMoPesa from a company called MoMoCom. I am going to be showing you some of their processes, and asking you for your feedback. The purpose of the study is to help us evaluate the usability of the new processes. Remember we are not evaluating you as an agent in any way. We are evaluating the new mobile money process from MoMoCom. I also want to let you know that I don’t work for MoMocom, so your feedback won’t hurt my feelings.

\noindent \textbf{Standard Transaction Process\\}
\noindent \emph{Show a demo of the standard process that the customers usually follow.}
\begin{enumerate}
    \item Is the process I have just shown to you similar to the current process that customers currently follow when they are withdrawing money?
    \item Currently when you register users and when users are using mobile money, how do you know who they are? (probe: What do you use to confirm their identity?)
    \item Why is this confirmation important/necessary?
    \item How often do you need to verify people’s identity? 
    \item What do you think about using IDs to authenticate your customers? 
    \item Is there anything you would change about the authentication process using IDs  that you use at the moment?
    \begin{enumerate}
        \item Probe: are there any difficulties or challenges you experience from using physical ID to authenticate mobile money users?
    \end{enumerate}
    \item What information is contained in the current transaction messages you receive as a summary of the customer’s transaction?
    \item Your telco provider wants to change the content of the message you receive about the customer’s transaction by removing unnecessary information. They have come to you to know the following:
    \begin{enumerate}
        \item What information is necessary for your records that you would want to keep in the SMS you receive? How is it useful?
        \item What information do you think is unnecessary, or you could do without? How is the information  useful?
    \end{enumerate}
\end{enumerate}

\noindent \textbf{Privacy Preserving Process\\}
\noindent\emph{Show privacy preserving process and show both the face/finger and voice authentication processes on the smartphone and basic phone prototypes:}
\begin{enumerate}
\setcounter{enumi}{8}
    \item In the new process you have seen, customers will not need to show their ID when transacting. They will follow the process that I have just demonstrated to register and authenticate themselves. 
    \begin{enumerate}
        \item What do you think about using such a process where customers  authenticate themselves instead of relying on IDs? 
        \item Do you see any challenges using this method of authentication?
        \item If your MoMo provider gave you  a choice either to authenticate customers using their ID or to have customers use the new process to  authenticate themselves using a selfie, their fingerprint or voice, which one would you choose? (Why?)
        \item Is there any situation where you would prefer to use IDs? 
    \end{enumerate}
    \item Is there any situation where you think customers would prefer to still use their IDs instead of authenticating themselves this way? 

\end{enumerate}

\noindent \emph{Show redacted message and also show what the customer would receive:}

\begin{enumerate}
\setcounter{enumi}{10}
    \item What do you think of this transaction message that you would receive when a customer transacts?
    \begin{enumerate}
        \item Which message would you prefer? This one or the one you currently get? Why?
    \end{enumerate}
\end{enumerate}

\noindent \textbf{Proxy-Withdraw Privacy Preserving Process}\\
\noindent \emph{Before showing the proxy-withdraw process, ask:}
\begin{enumerate}
\setcounter{enumi}{11}
    \item Has any of your customers ever transacted and sent someone else to collect the money? If yes, how do customers normally do this? If no, imagine someone (a child or an adult) comes to your shop and says they have been sent by someone who happens to be your customer who has withdrawn some money using your agent number.
    \begin{enumerate}
        \item As the agent, how do (would) you know or verify who the sender is? Does the sender need to do anything to facilitate verification e.g., calling the agent to say who they will send? Sending someone mutually known? Does the agent call the supposed customer?
        \item As the agent, how do (would) you know that the person collecting the money is the right person who was sent? 
        \item Does (would) the person who has been sent need to provide any information? (If yes, what information?) 
        \item How about the sender. Do (would) they need to provide any information? (If yes, what information?)
        \begin{enumerate}
            \item How do (would) they provide this information to you?
            \item Do (would) you store this information for future use (say when the person sending someone else is your customer) or does the customer need to provide it all the time?
            \item (If (ii) above is yes), how do you store the information (probe: What exactly do you store?)
        \end{enumerate}
        \item How does the customer get the agent number to withdraw the money given they are not at your shop when they are transacting?
        \item At what point do (would) you decide to actually give the person who has been sent the money. (Probe: What will make you confident to give this other person the money?)
    \end{enumerate}
    \item As an agent, is  there anything you like or dislike about this process where customers transact away from your kiosk and send other people to collect the money?
    \begin{enumerate}
        \item Is it possible to give money to the wrong person?
        \item Have you or another agent ever given money to the wrong person when someone says they have been sent?
    \end{enumerate}
\end{enumerate}
\noindent \emph{After showing  the proxy withdraw process, ask:}
\begin{enumerate}
\setcounter{enumi}{13}
    \item When customers are sending other people to collect money, they can either provide the person they are sending or to you specific information about the transaction, or they can use the process I have just shown you. As an agent, given the option to educate customers on one of these methods, which one would you point them to? Why? 
    \item Is there anything you like about such a MoMo  process where a customer can send someone else?
    \item Is there anything you dislike about it?
    \item Are there any challenges you foresee in such a process?
\end{enumerate}

\noindent \textbf{Other Questions}
\begin{enumerate}
\setcounter{enumi}{17}
    \item Are you aware of data privacy laws in Kenya?
    \item (If yes), What are your responsibilities as a mobile money agent as defined by these laws? 
\end{enumerate}

\section{Interview Script - Users} \label{app:interview-customers}

\noindent \textbf{Introduction}\\
\noindent Thank you for being here. We are in the process of testing mobile money transfer processes for a new mobile money product called MoMoPesa from a company called MoMoCom. I am going to be giving you some activities to do, and asking you for your feedback as you work on these tasks. The purpose of the study is to help us learn more about the usability of the new processes from MoMocom. The tasks are not in any way a test of your skills. So just do the best you can and if there are things that don’t make sense to you just let us know.\\
You will be working on a mobile phone which happens to be on paper. This is \emph{name of assisting researcher} and she will be playing as our mobile phone today and handing you the different screens based on your actions. And this is \emph{name of assisting researcher} and she will be playing the role of an agent today. So imagine that this is your phone area and your screens will be placed here. Use this pen to interact with the phone. Point to things you will normally select, and where you would normally type something, just tell me what you would type for example ``I will enter my phone number'' which is ``then say the phone number.'' MoMoPesa is the picture on your screen that looks like a wallet. So why don’t we try this:

\begin{itemize}
    \item Show me what you would select to open the MoMoPesa menu.
    \item Would you show me the option for saving for your business?
    \item Please show me now what you would select if there was an option to identify yourself on the screen.
    \item Now show me what you would do to call a number: 0770800900 that is not saved in your phone book.
\end{itemize}

\noindent \emph{[Researcher’s note: These are a test that the participant can read. Observe to make judgments on this. Only continue if participant passes tests]}

\noindent Thank you. For the purpose of the tasks today, we will assign you a name, a phone number and an ID number to use wherever these will be needed. Here is your ID. Please keep it where you other ID is. (Researchers note: Ensure the participant has kept this ID where their national ID is) These are your other details (Researcher to hand the participant the paper with their created profile). While working on the tasks you might encounter areas that we do not have a screen for, and that’s okay, we will guide you through it. As you work on your activities, please tell us what you are thinking and what’s going through your mind. For example, you might say, “I was expecting a different screen.” “I am feeling confused”  “this is different from what I am used to” or any other such thoughts that may come to your mind. All of this information is important for us to know.  Remember we are not evaluating you in any way, we are evaluating the new mobile money process from MoMoCom. I also want to let you know that I don’t work for MoMocom, so you won’t hurt my feelings if you don’t like something. During the activities, please imagine that you are transacting as you normally would in your regular life. When you finish a task, please let me know you are done. You may ask questions while doing the activity, however, I may not be able to answer all of them until the end of the session. Do you have any questions? \\
\emph{[Researcher’s note: Test that the participants understand the role play] 
Please take a look at the paper profile I gave you. Remember you are going to pretend to be this person today. Tell me, what is your name today? And \emph{fake study name?} What is your phone number? ID number?
}

\noindent Before we start, I would like to show you how I would think aloud when I want to call someone called Eko Kolipo on my phone.

\noindent \textbf{Task 1: Participant Withdraws Using Standard Current Process}\\
\noindent \textbf{Instructions and Task}\\
In this first task, you want to withdraw Ksh 500  from  an agent. Please complete this process as you would using mobile money. \emph{Name} will be your agent. Interact with them as you would when you go to an agent. \noindent \emph{[Researchers note: During each task, ask the participant to repeat the task they are required to complete. If the participant is not thinking aloud, prompt with things like: Tell me what you are thinking about right now.  What are you looking at? What are you looking for? What are you trying to decide?]}

\noindent \textbf{Interview Questions}
\begin{enumerate}
    \item What were your impressions of this process? (prompt: what makes it \emph{their response}) \emph{[Researcher’s note: Also ask about any observations that were not addressed through the think aloud: I noticed that…]}
    \item Was the procedure you just completed  similar to the process you would usually follow when cashing out mobile money? 
    \begin{enumerate}
        \item (If no) What about it was different?
        \item Anything you particularly liked about this process?
        \item Anything you did not like?
    \end{enumerate}
    \item Why do you think you are asked to provide your ID when you transact at an agent?
    \item Have you ever been inconvenienced by the need to  show your ID when withdrawing mobile money?  (If yes) How so? 
    \item What information do you think the agent receives when you transact using the procedure you just followed? (prompt with: Did the agent receive your name? Your ID? Your phone number? How much money you wanted to transfer?)
    \begin{enumerate}
        \item Is there any situation where you would not want to share any of this information with the agent? (prompt for what and why)
    \end{enumerate}
    \item What information do you think is sent to \emph{the MoMo provider} when you initiate a withdrawal using the procedure you just followed?
    \begin{enumerate}
        \item Is there any situation where you would not want to share any of this info with the MoMo provider? (prompt: what and why?)
    \end{enumerate}
    \item Have you seen a sample message the agent gets when you transact?
\end{enumerate}

\noindent \textbf{Additional Questions - After Being Shown the Agent Notification}

\noindent I will now show you a sample agent notification message of your just completed transaction. \emph{[Researcher’s note: Show the Swahili or English message depending on what language participant uses for their MoMo]} 
\begin{enumerate}
\setcounter{enumi}{7}
    \item Consider the transaction summary that the agent received. 
    \begin{enumerate}
        \item What do you think about it?
        \item Did the content of the message surprise you?
        \item (If yes) What surprised you?
        \item (If no) Why were you not surprised?
        \item Would you change the information contained in the message?
        \item (If yes) How would you change it? ( What would you add/remove?) Why?
        \item (If no) Why would you not change anything?
        \item When transacting via MoMo, would you consider any of this information personal? (name, mobile number, amount, ID number, PIN, time of Transaction, balance)?
    \end{enumerate}
    \item Is there anything you would change about the process you just used? If yes, what and why?
\end{enumerate}

\noindent \textbf{Task 2: Participant Withdraws Using Privacy-Preserving Process}\\
\noindent \textbf{Advertisement Material}\\
Thank you \emph{fake name} for completing the first task. Before we move on to the other two tasks, I will show you a sample advertisement from MoMoCom explaining some of the features of their new product MoPESA, that will provide a different way to verify your identity and also allow a person to either withdraw cash themselves, or to withdraw and send someone else to collect. \emph{[Researcher’s note: Show ad(s) and demo] 
\begin{itemize}
    \item This is the advertisement showing the different ways you can use to be identified: \emph{Show the general ad with the three identity processes.}
    \item Now, I will show you how their registration of identification would work - In your case, it would be \emph{preferred biometric option}.
\end{itemize}}

\noindent \textbf{Comprehension Questions}
\begin{itemize}
    \item What are the ways you can identify yourself from the advertisement?
\end{itemize}

\noindent \textbf{Instructions and Task 2}\\
Now that you have an idea of the process from MoMoCom we will move to the second task. Assume that you already registered your preferred biometric: face or voice. Using this next process, I'd like you to withdraw Ksh 500  from an agent. This time, you are withdrawing yourself and not sending someone else. \emph{Name} will be your agent today. Interact with them as you would whenever you go to an agent.\\ \emph{[Researchers note: Ask the participant to repeat the task 
]}

\noindent \textbf{Interview Questions}\\
Thank you for completing the last tasks. I will now ask you some questions about this process that you have just used to withdraw money.

\begin{enumerate}
\setcounter{enumi}{9}
    \item What did you think of this process? (What makes it \emph{their response})
    \item Is there anything you liked about this process compared to the process of withdrawing money that you normally follow? (prompt for how e.g., if they say it was more secure)
    \item Is there anything you disliked about this process compared to the process that you normally follow? (prompt for how/what e.g., if they say it was difficult to use)
    \item If your provider (e.g., Safaricom) offered you a choice between this process and the one you use, which one would you choose? (Why?)
    \item Why do you think you did not have to provide your ID when completing the transaction using this process?
    \item Compared to the previous process where you had to show your ID:
    \begin{enumerate}
        \item Anything you liked about the new process of being identified?
        \item Anything you disliked about the new process of being identified?
        \item If your provider gave you a choice of using either your ID, or this new process of being identified, which one would you use? 
        \begin{enumerate}
            \item Why? What makes \emph{their option} better for you? 
        \end{enumerate}
    \end{enumerate}
    \item What information do you think the agent received when you transacted using this process? (prompt: Did the agent receive your name? ID? phone number? How much money you wanted to transfer?)
    \begin{enumerate}
        \item Is there any situation where you would not want to share any of this information with the agent?
        \item (If yes) which information would you not want to share, why? 
    \end{enumerate}
    \item What information do you think is sent to the mobile money provider when you complete the transaction using this process?
    \begin{enumerate}
        \item Is there any situation where you would not want to share any of this information with the mobile money provider?
        \item (If yes) Which information would you not want to share, why?
    \end{enumerate}
\end{enumerate}

\noindent Now I will show you a sample agent notification message  that the agent received when you transacted using this process (Note: Show message and show previous one for process 1)
\begin{enumerate}
\setcounter{enumi}{17}
    \item Consider the transaction summary that I have just shown you. 
    \begin{enumerate}
        \item Did you notice any difference in this message from the one you saw earlier?  (If yes), what was different?
        \item Do you think the notification provides the agent any way to be sure he is giving the money to the right person?
        \begin{enumerate}
            \item (If yes), how so? (No), what should be included? 
        \end{enumerate}
        \item What did you like or dislike about the content of the agent notification message?
        \item Would you change the info in this notification message?
        \item (If yes) How would you change it? (What would you add/remove?) Why?
        \item (If no) Why would you not change anything about it?
        \item Which message would you prefer that the agent receives about your transaction? Why?(prompt: How so where appropriate)
    \end{enumerate}
    \item Is there anything about this whole  process you would change? (prompt: Like anything you are concerned about, or something you feel could be better?)  If yes, what and why? 
\end{enumerate}

\noindent \textbf{Task 3: Proxy Withdraw Using Privacy-Preserving Process}\\
\noindent \textbf{Pre-Task Interview}\\
Thank you for completing the second task. We have one more task, but before you start, I would like to ask some questions to help me understand how people use MoMo. Sometimes, people send others to withdraw money for them and there are many different reasons why people would do this. 

\begin{enumerate}
\setcounter{enumi}{19}
    \item Have you ever sent someone else to withdraw money or to collect money that you have withdrawn? (If yes) Please tell me how you usually go about the process. (If not) Give scenario:  Imagine you were sick and could not get to the agent and so you needed to send someone; how would you go about this?
    \begin{enumerate}
        \item How would you decide about whether you would send someone or wait to do it yourself? (prompt: amount of money, location, a new agent vs known agent, urgency, collector is the receiver?)
        \item Who would you send? 
        \item What would (or did) you have to do to ensure that the person collects the money successfully? (prompt: What would (or did) you provide to the person you are sending?)
        \begin{enumerate}
            \item Ask for all mentioned: Why would the info be useful?
            \item (For information) How would (or did) you provide  this info? (prompt: do you write down, do you call?) 
            \item Is there any situation where you would not want to share some or all of the information you mentioned with the person you were sending? 
            \item (If yes) Which information and why?
        \end{enumerate}
        \item What would (or did) you have to provide  to the agent to ensure the person you sent collects the money without any problems?
        \begin{enumerate}
            \item How would (or did) you provide this information?
            \item Is there any situation where you would not want to share some or all of the info you mentioned with the agent? 
            \item (If yes) Which information and why?
        \end{enumerate}
    \end{enumerate}
    \item Is there anything you wish you could change about the way you currently send someone?
\end{enumerate}

\noindent \emph{[Researcher’s note: I will show you one more ad from Momocom that shows the different ways you can withdraw using their new process}\\
\noindent \textbf{Comprehension Questions}
\begin{itemize}
    \item What are ways you can withdraw money from the advertisement?
\end{itemize}

\noindent \textbf{Task 3: Instructions and Task}\\
\noindent \textbf{Part 1:} This is our third and last task. During this task, you might encounter some screens where we use X to hide part of the phone number and only show the last 4 digits like this (show sample screen). Now Imagine you want to withdraw Ksh 500 and send someone else to collect it for you from the agent. I would like you to initiate the transaction. Remember to keep telling me what you are doing. I will act as the person you would have sent in real life.[My name is Mimi Koleta and this is my phone number.] \emph{Name} will continue to play the role of the agent. \emph{[Researchers note: Ask the participant to repeat the task they are required to complete.]}

\noindent \textbf{Part 2:}
Now \emph{fake name}, I would like you to imagine that I am your friend and I would like to send you to go and collect for me Ksh 1000 that I have withdrawn. I have already initiated the transaction as you just did with your own transaction a few minutes ago. Now imagine that you are already at the agent and complete the cash collection for me. \emph{[Researchers note: Ask the participant to repeat the task they are required to complete.]}

\noindent \textbf{Interview Questions}\\
Thank you. I will now ask you some questions about your experience with these two processes of sending someone and collecting for someone.
\begin{enumerate}
\setcounter{enumi}{21}
    \item What did you think of the process?
    \item Is there anything you liked about this process compared to the current one that people use  when sending someone else to collect?
    \item Anything you disliked about this process compared to the current one that you (would) use when sending someone else to collect?
    \item Is there anything you found difficult? (if yes, what and how so?)
    \item Is there anything you found confusing? (if yes, what and how so?)
    \item If you were using this process:
    \begin{enumerate}
        \item What would you consider more carefully compared to when you send someone the normal way? 
        \item What would you consider less carefully (or would be less important) compared to when you send someone the normal way
        \item Would your decision of who you send change because of using this process? (If yes) how so?
    \end{enumerate}
    \item Are there any reasons why you would not use this process? 
    \item What info did  the person you were sending need to have or see when you sent them to collect using the procedure you just followed?
    \begin{enumerate}
        \item Is there any situation where you would not want to share some or all of the information you mentioned with the person you were sending? (If yes) What and why not?
    \end{enumerate}
    \item What information do you think the agent received when you transacted using this last procedure that you just followed?
    \begin{enumerate}
        \item Is there any situation where you would not want to share any of this information with the agent? (Note: correct participant after they respond if they have a wrong perception of what was sent)
    \end{enumerate}
    \item What information do you think was sent to the mobile money provider when you transacted using this procedure?
    \begin{enumerate}
        \item Are you comfortable with the mobile money provider having access to this information? Why or why not?
        \item Is there any situation where you would not want to share any of this information with the mobile money provider?
    \end{enumerate}
    \item If you were to send someone to collect, would you want them to use this new process of collection or would you prefer the current method?Why?
    \item If you were to collect money for someone else, would you want to use this new process of collection? Why? 
    \item Overall, if you needed to withdraw money and send someone else to collect and your MoMo provider offered you this process as an option, which one would you choose – Would you use this process or would you prefer not to use this process at all?  Why?
    \item Is there anything about this process you would change?(Additional probe: Like anything you are concerned about, or something you feel could be better).  If yes, what and why?
\end{enumerate}


\section{Additional Figures}
\label{app:additional}

\begin{table}[H]
\centering
\small
\footnotesize
\caption{Participant demographics for users and agents}
\resizebox{\linewidth}{!}{
\begin{tabular}{l c c c c}
\toprule
 & \multicolumn{2}{c}{\textbf{Users}} & \textbf{Agents} \\
 & \textbf{Smartphone} & \textbf{Feature phone} & \\
\midrule
\textbf{Gender} & & & \\
Male & 10 & 7 & 8 \\
Female & 9 & 6 & 7 \\
\midrule
\textbf{Age} &  &  & \\

18-24 & 5 & 2 & 3 \\
25-34 & 10 & 7 & 6 \\
35-44 & 4 & 3 & 2 \\
45+ & 0 & 1 & 3 \\
\midrule
\textbf{Education} &  &  & \\
Below high-school & 0 & 2 & 0 \\
High-school & 6 & 6 & 8 \\
Post high-school & 3 & 5 & 7 \\
\bottomrule
\label{table:demographics}
\end{tabular}}
\label{tbl_demographics}
\end{table}


\balance
 
\figureDeposit

\figurePrivacyDelegated

\clearpage

\begin{table*}
    \begin{tabular}{p{0.08\textwidth}|p{0.2\textwidth}|p{0.2\textwidth}|p{0.2\textwidth}|p{0.22\textwidth}}
    \hline
          Party & Fraud and Theft (n=26) & Unauthorized Use (n=16) & Physical Harms (n=10) & Economic Judgments (n=6) \\
    \hline
    \hline
         Agent & PIN, ID No., Account Balance, Phone Number, Amount Withdrawn & ID No., Phone No., Name, Amount Withdrawn & Amount Withdrawn, Account Balance, Phone No., PIN & Account Balance, Amount Withdrawn \\
         Proxy & PIN, Phone & - & - & Account Balance, Amount Withdrawn \\
         Provider & PIN, ID No., Phone No., Name, Account Balance & - & - & Account Balance, Amount Withdrawn\\
    \hline
    \multicolumn{5}{c}{\textbf{(n): number of respondents who mentioned}} \\
    \hline
    \end{tabular}
    \caption{A summary of some common data concerns and the related threat actor from the perspective of users}
    \label{tab:user_concerns}
\end{table*}

\begin{table*}
    \vspace{-1cm}
    \begin{tabular}{p{0.1\textwidth}|p{0.2\textwidth}|p{0.2\textwidth}|p{0.2\textwidth}|p{0.2\textwidth}}
    \hline
          Topic & Theme & Subtheme & Example User Codes & Example Agent Codes\\
    \hline
    \hline
         Data Sharing Perceptions and Attitudes & Privacy Perceptions & What is Private and Why & ID connects to all details, Conditional sensitivity & - \\
         \cline{3-5}
         &  & Privacy from Whom? & Access by agent/provider, Sharing with proxy & - \\
         \cline{3-5}
          & & Risks That Concern Users & Economic status judgments, Fraud/theft, Physical harm, Unauthorized use by agents & - \\
          \cline{3-5}
          & & Risk Mitigation Techniques & Blur PII/info in process, Add partial ID to message & - \\
          \cline{2-5}
          & Factors Influencing Data-Sharing & Perceived Utility of Shared Data & Authentication/verification, Transaction monitoring & Required for transaction identification/completion\\
          \cline{3-5}
          & & Perceptions about Data Sensitivity & Low sensitivity, Nothing to hide, Privacy-utility tradeoffs & - \\
          \cline{3-5}
          & & Data Sharing Based on Trust & Trusted proxy/agent/provider & - \\
          \cline{3-5}
          & & Obligatory Compliance & Forced data disclosure & Provider requirement\\
    \hline
          Existing Process Inconveniences & ID-Related Challenges & Usability Challenges & Call agent beforehand, ID-dependent access & Sender-agent communication, ID-dependent KYC\\
          \cline{2-5}
          & Inconvenience of Balancing Data-Sharing and Security & - & Sharing personal identification, Change PIN after transaction, Auditing after proxy & - \\
          \cline{2-5}
          & Contextual Complexity in Decision Making & - & Age of sendee, Transaction amount, Location & Flexible ID/transaction practices, Relationship/trust\\
    \hline
          Impressions of Proposed Protocols & Better Privacy and Security Affordances & Better Security & Not sharing ID, Code expiration & Better proxy verification, Deters scammers\\
          \cline{3-5}
          & & Better Privacy & Sender controls access, Does not share personal information & Protects customer information\\
          \cline{3-5}
          & & Better Usability & Confirm transaction, Convenient for all parties & Free agent from storing customer's information, Authentication convenient for all parties\\
          \cline{2-5}
          & Usability and Privacy Concerns & Inconvenient Process & Cumbersome process, Biometric concerns, Many authentication steps & Lack simplicity, Many authentication steps\\
          \cline{3-5}
          & & Complex Protocols & Details to enter, Human error concerns, Learn new process & Many steps for proxies, Code non-expiry, Workflow confusion\\
          \cline{3-5}
          & & Concerns with Voice Biometric & Biometric concerns & Biometric concerns\\
          \cline{3-5}
          & & Access and Accessibility Issues & Restricts sender & Restricts sender\\
    \end{tabular}
    \caption{A condensed version of the codebook with themes, subthemes, and example codes for both users and agents.}
    \label{tab:codebook}
\end{table*}

	\section{Security and Privacy Analysis}
    \label{app:analysis}
    \normalsize
	In this section, we provide a basic analysis of the security and the privacy of our base privacy-preserving protocol (i.e., the protocol without delegated withdrawals). 
	
	We let $\mathcal U$ denote the set of all users, and we denote the owner of account $a$ by $\Omega(a) \in \mathcal U$. We will consider events in which a user $u$ initiates a transaction (withdrawal or deposit) from a mobile money account $a$ for amount $x$ at time $t$. 
	
	To analyze the privacy properties of our protocol, we model each person as being characterized by a set of attributes $v=\{v_1, \ldots, v_n\}$. 
	A subset of these attributes are \emph{side information} accessible to the agent $v_S \subseteq v$. 
	These could include visible attributes (e.g., hair color) as well as prior information the agent holds about people. 

	\paragraph{Assumptions}
	To formalize our guarantees, we make the following assumptions:
	\begin{enumerate}
		\item \textbf{Threat model:} 

        \noindent \emph{Agent:} We assume the agent is honest-but-curious. That is, the agent follows protocol, but tries to infer information about the user. 
        
        \noindent \emph{User:} We model a user that can violate protocol in the following ways: (1) 
        An honest user (i.e., a person who follows protocol) may lose control of their phone. This can happen either before a withdrawal is attempted, or during the withdrawal process, between the authentication and the physical withdrawal. 
        (2) A dishonest user may attempt to maliciously withdraw money. There are two main ways for a user to behave maliciously---they can attempt to withdraw from an account that is not theirs, or they can attempt to collect more cash than they withdrew on the application. 

        \noindent \emph{Collector:} In the case of delegated withdrawals, we further assume that a collector can break protocol in a number of ways. Both honest and malicious collectors can behave (or be compromised) in the same ways as the user above. 
        Specifically, for malicious collectors, they can attempt to collect a transaction that they were not authorized to collect, or they can attempt to collect more cash on a transaction they were authorized to collect. 

		\item \textbf{Security of Biometric Authentication:} We assume a user $u\in \mathcal U$ 
		can authenticate themselves (using the biometric authentication application) 
		 if and only if they are the rightful owner of the MoMo account. 
		\item \textbf{Network latency} We assume that protocol messages are transmitted instantaneously between parties.
		\item \textbf{Random Code Generation} We assume that the random code $C\in \mathbb B^{\ell}$ is generated independently of the user's identity and the agent's side information. 
		Specifically, for a random code $C$, a random user $U$, and a random side information vector $V_s$, we observe that 
		$$\mathbb P(C=c~|~U=u,V_S=v_S) = \mathbb P(C=c),$$
		and $C\sim \text{Unif}(\{0,1\}^\ell)$.
        \item \textbf{Software Integrity} We assume that the software on all devices (user, collector, agent, MoMo provider) is uncompromised and runs according to protocol. 
	\end{enumerate}

	\subsection{Security}
	We first prove that under our assumptions, incorrect withdrawals (e.g., stemming from a stolen phone or an attempt to withdraw from someone else’s account) cannot happen in the privacy-preserving withdrawal protocol. This shows that our protocol does not degrade the  integrity of withdrawals.
	
	\begin{prop}[Security of Privacy-Preserving Withdrawal]
		Consider a withdrawal event under Protocol \ref{fig:priv_cashout}, in which user $u$ initiates a withdrawal  from a MoMo account $a$ for amount $x$, and $u$ conducts the withdrawal authentication exactly once. 
		After authenticating as the owner of account $a$, user $u$ obtains a code $c$ at time $t$. 
		Next, a  user $u’\in \mathcal U$ (it is possible but not necessary that $u=u'$) requests to collect  the withdrawn money $x$ from agent $g$ at time $t’>t$, using code $c$. 
		The agent will not give $x$ to $u'$ 
		if any of the following is true, except with negligible probability\footnote{We say that a function $f:\mathbb N \to \mathbb R$ is negligible if for any constant $\gamma>0$, there exists an integer $L_\gamma$ such that for all $\ell \geq L_\gamma$, $f(\ell)<\frac{1}{\ell^\gamma}$.} in $\ell$, the length of the code $c$:
		\begin{enumerate}
			\item $u\neq u'$ and $u$ has ownership of their phone at time $t'$.
			\item $u\neq u’$ and $u$ maintains control of their phone for at least $\Delta$ time after authentication, i.e. in the interval $[t,t+\Delta]$.
			\item $u=u'$ and $u \neq \Omega(a)$---that is, the user who initiates and executes the withdrawal is not the owner of the account. 
			\item $t’ > t+\Delta$---that is, the validity period of the random code expires.
			\item $x\neq x’$---that is, the amount of money disseminated is different from the amount requested during the withdrawal request on the mobile phone.
		\end{enumerate}
		\label{prop:security}
	\end{prop}
	
	\begin{proof}
		We consider each of the above cases in succession:
		\begin{enumerate}
			\item If $u\neq u'$ and $u$ has ownership of their phone at time $t'$, then at some time in the interval $[t'-\Delta, t']$,  user $u'$ must have guessed the code $c$ and amount $x$ exactly. 
			This can happen with probability at most $2^{-\ell}$ (Assumption 4). 
			\item If $ u\neq u’$ and $u$ maintains control of their phone for at least $\Delta$ time after authentication, then there are two options. 
			Either $u'$ acquired (e.g., stole)  user $u$'s phone after time $t+\Delta$, or $u'$ inferred the code $c$. 
			If $u'$ acquired user $u$'s phone after time $t+\Delta$, the random code would have expired (at time $t+\Delta$). Hence, the agent will not disseminate funds to $u’$, by Assumption 1. 
			If $u'$ inferred the code $c$, this happens with probability at most $2^{-\ell}\}$ (see proof of Case 1).
			\item If $u=u' \neq \Omega(a)$ then $u$ incorrectly authenticated themselves as $\Omega(a)$ (i.e., this is a failure of the biometrics system).	This cannot happen by Assumption 2.
			\item Again, we have two cases. If $u'$ uses $u$'s phone to withdraw (e.g., either because $u=u'$ or because $u'$ obtained $u$'s phone at time $t'$), then if $t’ > t+\Delta$, the random code has expired on $u$'s device. The agent follows protocol (Assumption 1), and the protocol only allows the agent to disseminate funds within a time window $\Delta$ of $t$, the time when the authentication is completed. 
			The second case is that $u\neq u'$, and at some time in the interval $[t'-\Delta, t']$,  the user $u'$ obtained the code $c$ on their own device. 
			This can happen with probability at most $2^{-\ell}$. 
			\item If $x\neq x’$ and $u'$ withdraws from $u$'s phone, the agent will again not disseminate the amount because the protocol tells the agent to disseminate amount $x$ to the person whose phone displays the random code, and the agent follows protocol (Assumption 1).
			If $u'$ withdraws from a different phone but displays code $c$, this can happen with probability at most $2^{-\ell}$. 
		\end{enumerate}
		\end{proof}

        \begin{prop}[Security of Privacy-Preserving Delegated Withdrawal]
		Consider a delegated withdrawal event under Protocol \ref{fig:delegated_withdraw}, in which a sender $u$ initiates a delegated withdrawal from  account $a$ for amount $x$. User $u$ conducts the withdrawal authentication (i.e., authenticating as the owner of account $a$) exactly once. Further, $u$ designates a collector $w \in \mathcal U$ (identified by their phone number/account number $b$) to collect the withdrawal. At the time of withdrawal, collector $w$ first authenticates themselves as the owner of the account, i.e., $\Omega(b)$, and then obtains a code $c$ at time $t$. 

        Finally, a  user $u'\in \mathcal U$ (it is possible but not necessary that $u'=w$) requests to collect  the withdrawn money $x$ from agent $g$ at time $t’>t$, using code $c$. 
		The agent will not give $x$ to $u'$ 
		if any of the following is true, except with negligible probability in $\ell$, the length of the code $c$:
		\begin{enumerate}
			\item $w\neq u'$ and $w$ has ownership of their phone at time $t'$.
			\item $w\neq u’$ and $w$ maintains control of their phone for at least $\Delta$ time after authentication, i.e. in the interval $[t,t+\Delta]$.
			\item $w=u'$ and $w \neq \Omega(b)$---that is, the collector who initiates and executes the collection is not the owner of the delegated collector's account, $b$. 
			\item $t’ > t+\Delta$---that is, the validity period of the random code expires.
			\item $x\neq x’$---that is, the amount of money disseminated is different from the amount requested during the withdrawal request on the mobile phone.
                \item $w=u'$, but $u$ did not complete the delegation process to $w$ (account $b$).
		\end{enumerate}

		\label{prop:security-delegated}
	\end{prop}
	
	\begin{proof}
        Cases 1-4 follow the same logic as Proposition \ref{prop:security}.
        For case 5, if $u$ does not complete the delegation process, then the MoMo provider does not provide user $w$ with a delegated withdrawal transaction from $u$, by Protocol \ref{fig:delegated_withdraw}. 
        In this case, if $w$ tries to withdraw money from account $a$, it will be unable to select such a transaction from their delegated transaction list. Hence, $w$ would have to directly initiate a withdrawal from user $u$'s account $a$, which can happen only with negligible probability by Proposition \ref{prop:security}.

		\end{proof}
	
	\subsection{Privacy}

	
	Our protocol protects privacy in the sense that it promotes data minimization---it reveals less information by default to the agent than they would have otherwise seen \cite{pfitzmann2010terminology}. However, it is not private in a cryptographic sense.
	In the worst-case, an agent may be able to learn more about a user from the protocol than they would have without running the protocol. Consider the following example: 
	An agent has side information that everyone in their village is of ethnicity $E_1$. However, the agent knows that there is a very rich merchant named Mohamed of ethnicity $E_2$ who occasionally visits their village to sell goods.
	Hence, if an agent sees a person of ethnicity $E_2$ and the person executes the deposit protocol to deposit a large sum of money, the agent can infer from the side information and the transaction amount that the person is probably Mohamed.
	If the customer had not deposited a large sum, the agent would have less certainty that the person is Mohamed. 
	In this sense, the protocol leaks information.

	We can reason about the amount of information leaked by the protocol. 
	For a given agent, consider a random variable $U$ which represents a random user drawn from the set of users $\mathcal U$ according to some arbitrary distribution that is specific to the agent (in our earlier notation, $u$ is a given realization of random variable $U$). 
	For instance, if the agent lives in a remote village, the distribution over $\mathcal U$ may place much heavier probability mass on the members of that village. 
	
	Similarly, we define random variables $V_S$, representing the side information associated with a user, and $T$ representing the random transcript that the agent sees from a given execution of the protocol. 
	Intuitively, the protocol leaks information if the agent learns more about a user than they would by only observing the user's side information $V_S$. 
		\begin{prop}
			The mutual information between $U$ and the agent's view of the transcript (represented by random variable $T$), conditioned on the side information random variable $V_S$ satisfies the following:
			$$
			I(U;T,V_S) - I(U;V_S) = I(U,T|V_S).
			$$
            \label{prop:privacy}
		\end{prop}
		\begin{proof}
		We have that
		\begin{align}
			I(U;T,V_S) - I(U;V_S)  & = H(U) - H(U|T,V_S) - H(U) + H(U|V_S) \nonumber\\
			& = H(U|V_S) - H(U|V_S,T) \nonumber \\
			& = H(U|V_S) - H(U|C,X,V_S)  \label{eq:aux}  \\
			 &= H(U|V_S) - H(U|X,V_S)   \label{eq:indpt} \\						
						&=  I(U;T|V_S)  \label{eq:cond} 
		\end{align}
		First, we observe that in the transcript $T$, the only information sent to the agent is the random code $c$ and the transaction amount $x$ (which are realizations of random variables $C$ and $X$, respectively). 
		Any other information the agent learns can be learned without executing the protocol, and therefore falls under the side information $V_S$. 
		\eqref{eq:aux} follows because of this observation.
		\eqref{eq:indpt} follows because the code $C$ is independent of $U$ and $V_S$ (by Assumption 4). 
		Finally, \eqref{eq:cond} follows by the definition of conditional mutual information.
	\end{proof}
	This proposition says that by running the protocol, the information learned by the agent is the mutual information between the user's identity and the payment amount, conditioned on the side information. 
	This is because the code does not leak information, as it is independent of the user's identity and side information. 
	
    \noindent \textit{Remark:}
    To extend Proposition \ref{prop:privacy} to delegated withdrawal, we can replace the random variable $U$ with the joint random variable $(U,W)$, which represents the (sender, collector) pair. Under this substitution, the same proof logic from Proposition \ref{prop:privacy} holds.
    



\end{document}